\documentclass[english,onecolumn,draftcls]{IEEEtran}

\usepackage[T1]{fontenc}
\usepackage[utf8]{inputenc}
\usepackage{amsthm}
\usepackage{amsmath}
\usepackage{amssymb}
\usepackage{graphicx}
\usepackage{dsfont}
\usepackage{cite}
\usepackage{amsmath}
\usepackage{array}
\usepackage{multirow}
\usepackage{caption}
\usepackage{color}

\usepackage{cite}

\usepackage{amsmath,amssymb,amsthm,dsfont,mathrsfs}  % improve math presentation
\usepackage{dsfont}
\usepackage{graphicx} % takes care of graphic including machinery
\usepackage{hyperref}
\hypersetup{colorlinks=true,linkcolor=blue,citecolor=blue,linktocpage}

\newcommand{\E}{\mathbb{E}}
\newcommand{\Var}{\mathrm{Var}}

\newtheorem{theorem}{Theorem}
\newtheorem{corollary}{Corollary}
\newtheorem{lemma}{Lemma}

\theoremstyle{definition}
\newtheorem{definition}{Definition}

\usepackage{algorithm}
\usepackage{algorithmic}

\usepackage{fullpage}

\begin{document}

\title{Improved Bounds and Algorithms for \\ Sparsity-Constrained Group Testing}

\author{Nelvin Tan and Jonathan Scarlett \thanks{The authors are with the  Department of Computer Science, National University of Singapore (NUS).  J.~Scarlett is also with the Department of Mathematics, NUS.  e-mail: \url{nelvintan@u.nus.edu};  \url{scarlett@comp.nus.edu.sg}

This work will be presented in part at the 2020 IEEE International Symposium on Information Theory (ISIT) \cite{Nel20}.

This work was supported by an NUS Early Career Research Award.}}

%\author{Tan Thong Cai, Nelvin}
%\projyear{2019/2020 Semester 2}
%\projnumber{U240020}
%\advisor{Prof. Jonathan Scarlett}
%\deliverables{
%	\item Report: 1 Volume
%	}

\allowdisplaybreaks

\maketitle

\begin{abstract}
    In group testing, the goal is to identify a subset of defective items within a larger set of items based on tests whose outcomes indicate whether any defective item is present. This problem is relevant in areas such as medical testing, data science, communications, and many more. Motivated by physical considerations, we consider a sparsity-based constrained setting (Gandikota et al., 2019) in which the testing procedure is subject to one of the following two constraints: items are finitely divisible and thus may participate in at most $\gamma$ tests; or tests are size-constrained to pool no more than $\rho$ items per test.
    While information-theoretic limits and algorithms are known for the non-adaptive setting, relatively little is known in the adaptive setting. We address this gap by providing an information-theoretic converse that holds even in the adaptive setting, as well as a near-optimal noiseless adaptive algorithm for $\gamma$-divisible items. In broad scaling regimes, our upper and lower bounds asymptotically match up to a factor of $e$. We also present a simple asymptotically optimal adaptive algorithm for $\rho$-sized tests.
    In addition, in the non-adaptive setting with $\gamma$-divisible items, we use the Definite Defectives (DD) decoding algorithm and study bounds on the required number of tests for vanishing error probability under the random near-constant test-per-item design. We show that the number of tests required can be significantly less than the Combinatorial Orthogonal Matching Pursuit (COMP) decoding algorithm, and is asymptotically optimal in broad scaling regimes.
\end{abstract}

\section{Introduction} \label{ch:introduction}

In the group testing problem, the goal is to identify a subset of defective items of size $d$ within a larger set of items of size $n$ based on a number $T$ of tests.  We consider the noiseless setting, in which we are guaranteed that the test procedure is perfectly reliable: We get a negative test outcome if all items in the test are non-defective, and a positive outcome outcome if at least one item in the test is defective.  This problem is relevant in areas such as medical testing \cite{Dor43}, data science \cite{Gil08}, communication protocols \cite{Ant11}, and many more \cite{Ald19}. 
One of the defining features of the group testing problem is the distinction between the non-adaptive and adaptive settings. In the non-adaptive setting, all tests must be designed prior to observing any outcomes, whereas in the adaptive testing, each test can be designed based on previous test outcomes. 

While the majority of the group testing literature has allowed for arbitrary (i.e., unconstrained) test designs, these may be unrealistic in many practical scenarios.  To address this, a {\em sparsity-constrained} group testing setting was recently proposed in \cite{Ven19}, in which the tests are subject to one of two constraints: (a) items are \textit{finitely divisible} and thus may participate in at most $\gamma$ tests; or (b) tests are \textit{size-constrained} and thus contain no more than $\rho$ items per test. These constraints are motivated by physical considerations, where each item has a limitation on the number of samples (\textit{e.g.}, blood from a patient) it can be divided into, or the testing equipment has a limitation on the number of items (\textit{e.g.}, volume of blood a machine can hold). The focus in \cite{Ven19} was on non-adaptive testing, and the two main goals of the present paper are the following: (i) provide a detailed treatment of the adaptive setting; (ii) close some notable gaps between the upper and lower bounds derived in \cite{Ven19} in the non-adaptive setting with $\gamma$-divisible items (in contrast, the gaps in \cite{Ven19} were less significant for $\rho$-sized tests).

\subsection{Problem Setup} \label{sec:setup}

Let $n$ be the number of items, which we label as $\{1,2,\dots,n\}$. Let $\mathcal{D}\subset \{1,2,\dots,n\}$ be the set of defective items, and $d=|\mathcal{D}|$ be the number of defective items.  We assume that $\mathcal{D}$ is generated uniformly at random among all sets of size $d$ (also known as the combinatorial prior\cite{Ald19}).\footnote{Despite this assumption, our adaptive algorithm attains zero error probability (see Theorem \ref{thm:gamma_upperbound_adap}), thus ensuring success even for the worst case $\mathcal{D}$ of cardinality $d$.}
% We write $u_i=1$ to denote that item $i\in\mathcal{D}$ is defective, and $u_i=0$ to denote that $i\notin\mathcal{D}$ is non-defective. In other words, $u_i$ is the indicator function $u_i=\mathds{1}\{i\in\mathcal{D}\}$. We then write $\mathbf{u}=(u_i)\in\{0,1\}^n$ for the defectivity vector. 

We are interested in asymptotic scaling regimes in which $n$ is large and $d$ comparatively is small, and thus assume that $d = o(n)$ throughout.  We let $T=T(n)$ be the \textit{number of tests} performed and label the tests $\{1,2,\dots,T\}$. To keep track of the design of the test pools in the non-adaptive setting, we write $x_{ti}=1$ to denote that item $i\in\{1,2,\dots,n\}$ is in the pool for test $t\in\{1,2,\dots,T\}$, and $x_{ti}=0$ otherwise. This can be represented by the matrix $\mathsf{X}\in\{0,1\}^{T\times n}$, known as the \textit{testing matrix} or \textit{test design}.

Similar to most theoretical studies of group testing \cite{Ald19}, our study of the non-adaptive setting will focus on random designs. % In this case, we write $X_{ti}$ to denote the random $(t,i)$-th entry of the random testing matrix.  
The most commonly considered test design in the unconstrained setting is the Bernoulli design (i.e., $\mathsf{X}$ has i.i.d.~Bernoulli entries); however, this creates significant fluctuations in the number of tests per item, which is undesirable in the case of $\gamma$-divisible items.  Hence, we will pay particular attention to the {\em near-constant tests-per-item} design \cite{Joh16,Coj19a}, in which each item is included in some fixed number $L$ of tests, chosen uniformly at random with replacement. Since we are selecting with replacement, some items may be in fewer than $L$ tests, hence the terminology ``near-constant''. This is a mathematical convenience that makes the analysis more tractable.

%In this paper, we are interested in the following design:\\
%5\textbf{Bernoulli design:} In this design, each item is included in each test independently at random with some fixed probability $p$. That is, we have $\mathbb{P}(X_{ti}=1)=p$ and $\mathbb{P}(X_{ti}=0)=1-p$ i.i.d. over $i\in\{1,2,\dots,n\}$ and $t\in\{1,2,\dots,T\}$. Typically the parameter $p$ is chosen to scale as $p=\Theta(1/d)$.\\
% \textbf{Constant tests-per-item design:} In this design, each item is included in some fixed number $L$ of tests, with the $L$ tests for each item chosen uniformly at random, independent from the choices of all other items. Thus, we have independent columns of constant-weight in X. Typically the parameter $L$ is chosen to scale as $L=\Theta(T/d)$. In fact, it is often more mathematically convenient to analyse a similar design below.\\
%\textbf{Near-constant tests-per-item design:} In this design, 
% \textbf{Doubly regular design:} In this design, there is both a constant number $L$ of tests-per-item (column weight) and a constant number $m$ of items-per-test (row weight), with $\text{number of ones}=nL=mT$. Like before, we can also consider the ``near-constant'' versions for mathematical convenience. Again, $L=\Theta(T/d)$, so $m=\Theta(n/d)$, is a useful scaling.

% The above design is practical because the random matrix only contains $T\times n$ entries unlike standard random coding designs in channel coding that requires exponential storage and computation.

Let $y_t\in\{0,1\}$ be the \textit{outcome} of the test $t\in\{1,2,\dots,T\}$, where $y_t=1$ denotes a positive outcome and $y_t=0$ a negative outcome. Hence, we have $\mathbf{y}=(y_t)\in\{0,1\}^T$ for the vector of test outcomes. Using the OR (or disjunction) operator $\bigvee$, we have 
\begin{align}
    y_t=\bigvee\limits_{i\in\mathcal{D}}x_{ti}.
\end{align}
%or equivalently, using the definition of $u_i$ above, 
%\begin{align}
%    y_t=\bigvee\limits_{i}x_{ti}u_i,
%\end{align}
A \textit{decoding} (or detection) \textit{algorithm} is a (possibly randomized) function $\widehat{D}:\{0,1\}^{T\times n}\times\{0,1\}^T\rightarrow\mathcal{P}(\{1,2,\dots,n\})$, where the power-set $\mathcal{P}(\{1,2,\dots,n\})$ is the collection of the subsets of items. %Given the tests and their outcomes, the decoding algorithm outputs an estimate vector $\widehat{\mathbf{u}}\in\{0,1\}^n$, representing an estimate of the defectivity vector of the population.  
Denoting the final estimate by $\widehat{\mathcal{D}}$, the error probability is given by
\begin{align}
    P_e&=\mathbb{P}(\widehat{\mathcal{D}}\neq\mathcal{D}),
\end{align}
where the probability is taken over the randomness of the set of defective items, and also over the test design and/or decoding algorithm if they are randomized.

{\bf Testing Constraints.} As mentioned above, our focus in this paper is on the sparsity-constrained group testing problem introduced in \cite{Ven19}, in which the testing procedure is subjected to one of two constraints:
\begin{enumerate}
    \item Items are \textit{finitely divisible} and thus may participate in at most $\gamma$ tests.
    \item Tests are \textit{size-constrained} and thus contain no more than $\rho$ items per test.
\end{enumerate}
For instance, in the classical application of testing blood samples for a given disease \cite{Dor43}, the $\gamma$-divisible items constraint may arise when there are limitations on the volume of blood provided by each individual, and the $\rho$-sized tests constraint may arise when there are limitations on the number of blood samples that the machine can accept.

\subsection{Related Work} \label{sec:previous}

There exists extensive literature providing group testing bounds and algorithms for unconstrained group testing \cite{Mal78,Mal80,Du93,Ati12,Ald14a,Cha14,Joh16,Sca15b,Ald15,Coj19,Coj19a}; see \cite{Ald19} for a recent overview.  Here we focus our attention on those most relevant to the present paper.

In the absence of testing constraints, $T>(1-\epsilon)(d\log(\frac{n}{d}))$ tests are necessary to identify all defectives with error probability at most $\epsilon$ \cite{Cha14,Ald14a}.  Hence, the same is certainly true in the constrained setting.  The same goes for the {\em strong converse}, which improves the preceding bound to $T>(1-o(1))(d\log(\frac{n}{d}))$ for any fixed $\epsilon \in (0,1)$ \cite{Bal13,Joh15}.  A matching upper bound is known for all $d = o(n)$ in the unconstrained adaptive setting \cite{Hwa72}, whereas matching this lower bound non-adaptively is only possible non-adaptively in certain sparser regimes \cite{Sca15b,Coj19a}.

It is well-known that if each test comprises of $\Theta(\frac{n}{d})$ items, then $\Theta(d\log n)$ tests suffice for group-testing algorithms with vanishing error probability \cite{Cha14,Ald14a,Joh16}. Hence, the parameter regime of primary interest in the size-constrained setting is $\rho\in o(\frac{n}{d})$.  By a similar argument, 
%Similarly, if each item can be tested $\Theta(\log(\frac{n}{d}))$ times, then $\Theta(d\log n)$ tests suffice for group-testing algorithms with vanishing error probability \cite{Cha14,Ald16}. Hence, 
the parameter regime of primary interest for $\gamma$-divisible items is $\gamma\in o(\log(\frac{n}{d}))$.  Combined with the condition $T\in\Omega(d\log(\frac{n}{d}))$, the latter scaling regime implies that
\begin{align}
    \frac{T}{\gamma d} \to \infty \label{eq:T_gamma_d_ineq}
\end{align}
as $n \to \infty$, which will be useful in our proofs.

In the non-adaptive setting, Gandikota \textit{et al.} \cite{Ven19} proved the following results for $\gamma$-divisible items.

\begin{theorem}\textup{{\cite{Ven19}}} \label{thm:gamma_upperbound}
    For any sufficiently large $n$, sufficiently small $\epsilon>0$, $\gamma\in o(\log n)$, and $d\in\Theta(n^\theta)$ for some positive constant $\theta\in[0,1)$, there exists a randomized design testing each item at most $\gamma$ times that uses at most $\big\lceil e\gamma d(\frac{n}{\epsilon})^{1/\gamma}\big\rceil$ tests and ensures a reconstruction error of  at most $\epsilon$.
\end{theorem}

\begin{theorem}\textup{{\cite{Ven19}}} \label{thm:gamma_lowerbound}
    For any sufficiently large $n$, sufficiently small $\epsilon>0$, $\gamma\in o(\log n)$, and $d\in\Theta(n^\theta)$ for some positive constant $\theta\in[0,1)$, any non-adaptive group testing algorithm that tests each item at most $\gamma$ times and has a probability of error of at most $\epsilon$ requires at least $\gamma d(\frac{n}{d})^{(1-5\epsilon)/\gamma}$ tests.
\end{theorem}

For $\rho$-sized tests, the following achievability and converse results were also proved in \cite{Ven19}.
\begin{theorem}\textup{{\cite{Ven19}}} \label{thm:rho_upperbound}
    For any sufficiently large $n$, sufficiently small $\zeta>0$, $\rho\in\Theta\big((\frac{n}{d})^{\beta}\big)$ (for some constant $\beta\in[0,1)$), and $d\in\Theta(n^\theta)$ for some positive constant $\theta\in[0,1)$, there exists a randomized non-adaptive group testing design that includes at most $\rho$ items per test, using at most $\big\lceil\frac{1+\zeta}{(1-\alpha)(1-\beta)}\big\rceil\big\lceil\frac{n}{\rho}\big\rceil$ tests and ensuring a reconstruction error of at most $\epsilon=n^{-\zeta}$.
\end{theorem}
\begin{theorem}\textup{{\cite{Ven19}}} \label{thm:rho_lowerbound} 
    For any sufficiently large $n$, sufficiently small $\epsilon>0$, $\rho\in\Theta\big((\frac{n}{d})^{\beta}\big)$ (for some constant $\beta\in[0,1)$), and $d\in\Theta(n^\theta)$ for some positive constant $\theta\in[0,1)$, any non-adaptive group testing algorithm that includes $\rho$ items per test and has a probability of error of at most $\epsilon$ requires at least $\big(\frac{1-6\epsilon}{1-\beta}\big)\frac{n}{\rho}$ tests.
\end{theorem}
We observe that under the $\rho$-sized test constraint, both the lower and upper bounds have the same leading order term $\frac{n}{\rho}$. Hence, there is not much of a gap between the lower and upper bounds.\footnote{See also \cite{Geb20} for very recent improvements providing sharp constants.} However, under the $\gamma$-divisible items constraint, the lower bound contains the term $(\frac{n}{d})^{(1-5\epsilon)/\gamma}$ while the upper bound contains the term $(\frac{n}{\epsilon})^{1/\gamma}$. Hence, there is significant gap between the lower and upper bounds; we will see that the gap can be narrowed all the way down to a constant factor in the adaptive setting, and can also be significantly reduced in the non-adaptive setting.  See the following subsection for further details.  % In particular, for broad scaling regimes on $d$ and $\gamma$, and for the near-constant tests-per-item random design, we will provide upper and lower bounds on the number of tests that are both of the form $\Theta\big( \gamma d\big(\max\big\{d,\frac{n}{d}\big\}^{\frac{1}{\gamma}}\big)^{1+o(1)} \big)$, thereby matching asymptotically in the leading terms.

\subsection{Overview of the Paper}

The structure of the paper, as well as the main contributions, are outlined as follows:
\begin{itemize}
    \item In Section \ref{ch:sparse_adap_group_testing}, we consider the adaptive setting. We present an information-theoretic lower bound for $\gamma$-divisible items (Theorem \ref{thm:gamma_lower_bound_adap}), which strengthens the previous information-theoretic lower bound in \cite{Ven19} for $\gamma$-divisible items by improving its dependence on error probability, as well as extending its validity to the adaptive setting. Furthermore, we present adaptive algorithms for both $\gamma$-divisible items and $\rho$-sized tests, and show that both algorithms recover the defective set with zero error probability using a near-optimal number of tests (Theorem \ref{thm:gamma_upperbound_adap} and Section \ref{sec:adap_algo_rho}).  Informally, under mild assumptions, the optimal number of tests is shown to be within a factor $e^{1+o(1)}$ of $\gamma d(\frac{n}{d})^{1/\gamma}$ for $\gamma$-divisible items, and within a $1+o(1)$ factor of $\frac{n}{\rho}$ for $\rho$-sized tests.
    \item In Section \ref{ch:sparse_nonadap_group_testing}, we consider the non-adaptive setting. To further complement the preceding lower bound, we provide an additional lower bound that can be tighter (Corollary \ref{cor:gamma_lowerbound_ours_nearconstantwtdesign}), but that is specific to the near-constant tests-per-item design (rather than general designs).  In addition, we analyze the performance of the DD algorithm \cite{Ald14a} (Theorem \ref{thm:gamma_upperbound_nonadap}), and show that the number of tests can be significantly less than that of COMP (considered in \cite{Ven19}).  Informally, a special case of our results states that in the scaling regime $d = \Theta(n^{\theta})$ and $\gamma = \Theta( (\log n)^c )$ with $\theta,c \in (0,1)$, the optimal number of tests behaves as $\Theta(\gamma d\max\{n^\theta,n^{1-\theta}\}^{\frac{1+o(1)}{\gamma}})$.

    % By generalizing the $\gamma$-divisible items constraint and the $\rho$-sized tests constraint into a general constraint, we obtain information-theoretic lower bounds for both settings which are consistent with previous information-theoretic lower bounds in \cite{Ven19}. Furthermore, we extend an existing decoding algorithm to $\gamma$-divisible items and show that it performs better than a previously analyzed algorithm in .
    % \item In Section \ref{ch:conclusion_and_future_work}, we review the main contributions of the paper, and present various directions for future research.
\end{itemize}
In the final stages of preparing this paper, we noticed the concurrent work of \cite{Geb20}, whose results are similar to those that we develop for the non-adaptive setting.  In particular, the optimal number of tests for the near-constant tests-per-item design are characterized up to a constant factor in \cite{Geb20} whenever $\gamma = \Theta(1)$ (tight bounds are also given for $\rho$-sized test constraints, but these are more separate from our results).  While our results are not quite as strong in this regime (see the discussion following Theorem \ref{thm:gamma_upperbound_nonadap}), they have the advantage of also applying in regimes where $\gamma \to \infty$ (e.g., $\gamma = (\log n)^c$ for $c \in (0,1)$).  In addition, the proof techniques used are complementary, with ours building on \cite{Ald14} and \cite{Joh16}, whereas \cite{Geb20} builds on \cite{Coj19} and \cite{Coj19a}.  Finally, the adaptive setting is not considered in \cite{Geb20}.

{\bf Notation.} Throughout the paper, the function $\log(\cdot)$ has base $e$, and we make use of Bachmann-Landau asymptotic notation (\textit{i.e.}, $O$, $o$, $\Omega$, $\omega$, $\Theta$).

\section{The Adaptive Setting} \label{ch:sparse_adap_group_testing}

%\subsection{Introduction}
%
%In this section, we seek information-theoretic bounds and algorithms for the sparse adaptive setting. 
%% 
%Throughout the section, we consider the setup described in Section \ref{sec:setup} and Section \ref{sec:constraints}. More specifically, we consider the sparse adaptive group testing problem with small error probability, exact recovery, noiseless tests, binary outcomes, and combinatorial prior. Concretely, we target the error probability being bounded by some $\epsilon>0$:
%\begin{align}
%    P_e&=\mathbb{P}[\widehat{\mathbf{u}}\neq\mathbf{u}]\leq\epsilon,
%\end{align}
%where the probability is taken over the randomness of the set of defective items. Our main contributions are as follows:
%\begin{itemize}
%    \item In Section \ref{sec:adap_info_theo_lower_bound_gamma}, we provide an information-theoretic lower bound for $\gamma$-divisible items under the sparse adaptive setting.
%    \item In Section \ref{sec:adap_algo_gamma}, we provide an algorithm for $\gamma$-divisible items and study the number of tests for reliable recovery with zero error probability.
%    \item In Section \ref{sec:adap_algo_rho}, we provide an algorithm for $\rho$-divisible tests and study the number of tests for reliable recovery with zero error probability.
%\end{itemize}
%Our analysis will make use of several techniques and results from probability theory and information theory.

In this section, we seek information-theoretic bounds and algorithms for the adaptive setting, considering both the cases of $\gamma$-divisible items and $\rho$-sized tests.

\subsection{Information Theoretic Lower Bound for $\gamma$-divisible Items} \label{sec:adap_info_theo_lower_bound_gamma}

In this section, we present our information-theoretic lower bounds for sparse group testing under the $\gamma$-divisible items. We first prove a counting bound which gives us an upper bound on the success probability $\mathbb{P}(\textup{suc})=1-P_e$, following similar proof techniques as \cite{Bal13}, with suitable refinements to account for the $\gamma$-divisibility constraint. Afterwards, we will use the bound on $\mathbb{P}(\textup{suc})$ to prove the converse result (lower bound on $T$).
\begin{theorem} {\em (Counting-Based Bound)} \label{thm:counting_bound} 
    In the case of $n$ items with $d$ defectives where each item can be tested at most $\gamma$ times, any algorithm (possibly adaptive) to recover the defective set $\mathcal{D}$ with $T$ tests has success probability $\mathbb{P}(\textup{suc})$ satisfying
    \begin{align}
        \mathbb{P}(\textup{suc})\leq\frac{\sum_{i=0}^{\gamma d}{T\choose i}}{{n\choose d}}. \label{eq:prob_suc_bound}
    \end{align}
\end{theorem}
\begin{proof}
    See Section \ref{sec:counting_bound_proof}.
\end{proof}

The intuition behind \eqref{eq:prob_suc_bound} is that the denominator represents the number of defective sets of size $d$, and the numerator represents the number of possible tests outcomes (since there are always at most $\gamma d$ positive tests).  Once \eqref{eq:prob_suc_bound} is in place, the following converse follows from an asymptotic analysis.

\begin{theorem} \label{thm:gamma_lower_bound_adap} 
    {\em (General Converse Bound)} 
    Fix $\epsilon\in(0,1)$, and suppose that $d\in o(n)$, $\gamma\in o(\log{n})$, and $\gamma d\rightarrow\infty$ as $n\rightarrow\infty$. Then any non-adaptive or adaptive group testing algorithm that tests each item at most $\gamma$ times and has a probability of error of at most $\epsilon$ requires at least $e^{-(1+o(1))}\gamma d\big(\frac{n}{d}\big)^{1/\gamma}$ tests.
\end{theorem}
\begin{proof}
    See Section \ref{sec:gamma_lowerbound_adap_proof}.
\end{proof}

Since $\epsilon$ only affects the $e^{o(1)}$ term, asymptotically, the number of tests required remains unchanged for any nonzero target success probability.  This is in analogy with the strong converse results of \cite{Bal13,Joh15}. 

Theorem \ref{thm:gamma_lower_bound_adap} strengthens the previous information-theoretic lower bound in \cite{Ven19} for $\gamma$-divisible items (stating that $T\geq\gamma d(\frac{n}{d})^{(1-5\epsilon)/\gamma}$) by improving the dependence on $\epsilon$, as well as extending its validity to the adaptive setting (whereas \cite{Ven19} used an approach based on Fano's inequality that is specific to the non-adaptive setting).

\subsection{Adaptive Algorithm for $\gamma$-Divisible Items} \label{sec:adap_algo_gamma}

We first consider the recovery of the defective set given knowledge of the size  $d$ of the defective set. Afterwards, we consider the estimation of $d$.

\subsubsection{Recovering the Defective Set}

Our algorithm for the case that $d$ is known is described in Algorithm \ref{alg:adaptive_algo},
\begin{algorithm}[t]
    \begin{algorithmic}[1]
        \REQUIRE Number of items $n$, number of defective items $d$, and divisibility of each item $\gamma$
        \STATE Initialize $M\leftarrow (\frac{n}{d})^{\frac{\gamma-1}{\gamma}}$ and defective set $\mathcal{D}\leftarrow\emptyset$
        \STATE Arbitrarily group the $n$ items into $\frac{n}{M}$ groups of size $M$
        \STATE Test each group and discard any that return a negative outcome
        \STATE Label the remaining groups incrementally as $G^{(0)}_j$, where $j=1,2,\dots$
        \FOR{$i=1$ to $\gamma-1$}{
        % \STATE Initialize $j\leftarrow1$
        \FOR{each group $G^{(i-1)}_j$ from the previous stage}{
        \STATE Arbitrarily group all items in $G^{(i-1)}_j$ into $M^{1/(\gamma-1)}$ sub-groups of size $M^{1-i/(\gamma-1)}$
        \STATE Test each sub-group and discard any that return a negative outcome
        \STATE Label the remaining sub-groups incrementally as $G^{(i)}_j$
        }\ENDFOR
        }\ENDFOR
        \STATE Add the items in all the remaining groups $G^{(\gamma-1)}_j$ to $\mathcal{D}$
        \RETURN $\mathcal{D}$
    \end{algorithmic}
    \caption{Adaptive algorithm for $\gamma$-divisible items \label{alg:adaptive_algo}}
\end{algorithm}
where we assume for simplicity that $(\frac{n}{d})^{1/\gamma}$ is an integer.\footnote{Note that we assume $d \in o(n)$ and $\gamma\in o(\log(\frac{n}{d}))$, meaning that $(\frac{n}{d})^{1/\gamma} \to \infty$. Hence, the effect of rounding is asymptotically negligible, and is accounted for by the $1+o(1)$ term in the theorem statement.} Algorithm \ref{alg:adaptive_algo} is reminiscent of Hwang's generalized binary splitting algorithm \cite{Hwa72}, but the depth of the corresponding tree is controlled by using much more than two branches per split; see Figure \ref{fig:adaptive_algo}.

Using Algorithm \ref{alg:adaptive_algo}, we have the following theorem, which is proved throughout the remainder of the subsection.
\begin{theorem} \label{thm:gamma_upperbound_adap}
    {\em (Adaptive Algorithm Performance)} 
    For $\gamma\in o(\log n)$, and $d\in o(n)$, there exists an adaptive group testing algorithm that tests each item at most $\gamma$ times that uses at most $\gamma d(\frac{n}{d})^{1/\gamma}$ tests to recover the defective set exactly with zero error probability given knowledge of $d$.
\end{theorem}
\begin{proof}
    See Section \ref{sec:gamma_upperbound_adap_proof}.
\end{proof}

\textit{Comparisons:} Referring to Theorem \ref{thm:gamma_upperbound}, the upper bound for the non-adaptive algorithm of \cite{Ven19} using a randomized test design is $T\leq\big\lceil e\gamma d(\frac{n}{\epsilon})^{1/\gamma}\big\rceil$, where $\epsilon$ is the target error probability. The non-adaptive algorithm has a $(\frac{n}{\epsilon})^{1/\gamma}$ term in the upper bound, while our adaptive algorithm has a $(\frac{n}{d})^{1/\gamma}$ term. Since $\epsilon$ is small but $d$ is large, we see that our adaptive algorithm gives a significantly improved bound on the number of tests. Furthermore, the upper bound of our algorithm matches the information-theoretic lower bound in Theorem \ref{thm:gamma_lower_bound_adap} up to a constant factor of $e^{1+o(1)}$. This proves that our algorithm is nearly optimal.

\begin{figure}[t] 
  \centering
  \includegraphics[scale=0.35]{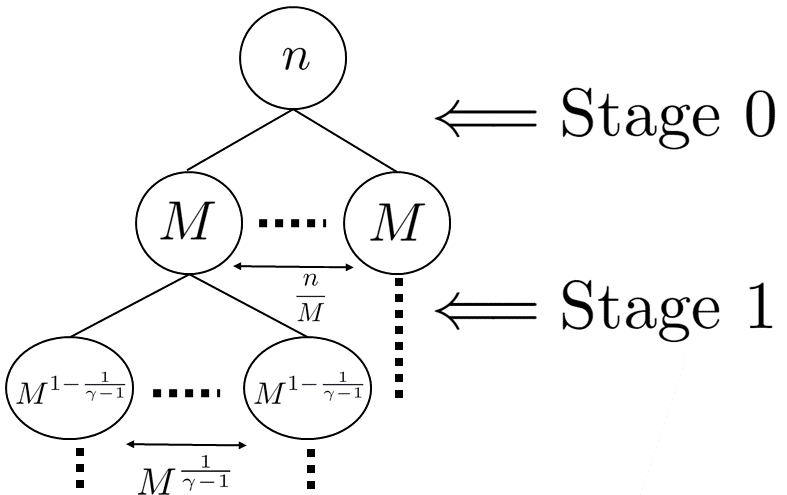}
  \caption{Visualization of splitting in the adaptive algorithm.}
  \vspace*{-2ex}
  \label{fig:adaptive_algo}
\end{figure}

\subsubsection{Estimating the Number of Defectives}

Since each item can appear in at most $\gamma$ tests, existing adaptive algorithms for estimating $d$ that place items in $\Omega(\log\log d)$ tests \cite{Dam10,Fal16} are not suitable when $\gamma \ll \log\log d$, and may be wasteful of the budget $\gamma$ even when $\gamma \gg \log \log d$.

To overcome this limitation, we introduce and evaluate two approaches to obtain a suitable input for $d$ in Algorithm \ref{alg:adaptive_algo} given knowledge of an upper bound $d_{\text{max}}\geq d$. The first approach uses $d_{\text{max}}$ directly in Algorithm \ref{alg:adaptive_algo}, while the second approach refines $d_{\text{max}}$ by deriving an estimate $\widehat{d}$ that is passed to Algorithm \ref{alg:adaptive_algo}. Note that we need $\widehat{d}$ to be an overestimate for the proof of Theorem \ref{thm:gamma_upperbound_adap} to still apply (with $\widehat{d}$ in place of $d$).

\paragraph{Using $d_{\text{max}}$ directly}

Assuming that $(\frac{n}{d_{\text{max}}})^{1/\gamma}$ is an integer, we first consider using $d_{\text{max}}$ directly in Algorithm \ref{alg:adaptive_algo} (in place of $d$) to recover the defective set $\mathcal{D}$.

\textit{Analysis:} Referring to Algorithm \ref{alg:adaptive_algo}, this changes our initialization of $M$ which becomes $(\frac{n}{d_{\text{max}}})^{(\gamma-1)/\gamma}$. Substituting the updated value of $M$ into \eqref{eq:upper_bound}, we obtain the following:
\begin{align}
    T\leq\frac{n}{(\frac{n}{d_{\text{max}}})^{(\gamma-1)/\gamma}}+(\gamma-1)d\Big[\Big(\frac{n}{d_{\text{max}}}\Big)^{\frac{\gamma-1}{\gamma}}\Big]^{\frac{1}{\gamma-1}},
\end{align}
which simplifies to 
\begin{align}
    T\leq(d_{\text{max}}-d+\gamma d)\Big(\frac{n}{d_{\text{max}}}\Big)^{\frac{1}{\gamma}}. \label{eq:upper_bound_method}
\end{align}

\paragraph{Binning Method}

\begin{algorithm}[t]
    \begin{algorithmic}[1]
        \REQUIRE Population of items, number of items $n$, upper bound $d_{\text{max}}\geq d$, and a probability parameter $\beta_n$
        \STATE Initialize number of bins $B\leftarrow d_{\text{max}}/\beta_n$
        \STATE Partition the items into $B$ bins of size $n/B$ each, uniformly at random
        \STATE Test each bin and discard any with a negative test outcome 
        \STATE $\widehat{d}\leftarrow\text{\#positive bins}/(1-\sqrt{\beta_n})$
        \RETURN $\widehat{d}$
    \end{algorithmic}
    \caption{Estimation of $d$ \label{alg:d_estimation_algo}}
\end{algorithm}
We will show that the bound on $T$ can be improved by forming a refined estimate of $d$ using knowledge of $d_\text{max}$, at the expense of having a non-zero (but asymptotically vanishing) probability of error.

Let $\beta_n$ be a given parameter, which we will assume tends to zero as $n \to \infty$. We first run Algorithm 2 to obtain a new input $\widehat{d}$ to Algorithm \ref{alg:adaptive_algo}. We then run Algorithm \ref{alg:adaptive_algo} with modified inputs (described in the following) to recover the defective set $\mathcal{D}$. Assuming that $(\frac{n}{\widehat{d}})^{1/\gamma}$ is an integer, we set the population of items in Algorithm \ref{alg:adaptive_algo} to be the remaining items left in the positive bins, the number of items as $d\times(\text{bin size})=d(\frac{\beta_n n}{d_{\text{max}}})$, the (upper bound on the) number of defective items as $\widehat{d}$, and the divisibility of each item as $\gamma-1$ (since each item is tested once in Algorithm \ref{alg:d_estimation_algo}).

\textit{Analysis:} We first show that the probability of a particular defective item colliding with any other defective item (\textit{i.e.}, falling in the same bin) tends to zero as $n\rightarrow\infty$. Referring to step 2 in Algorithm \ref{alg:d_estimation_algo}, conditioning on a particular item being in a particular bin, we see that the probability of another particular item being in the same bin is at most $1/B$. By the union bound, the probability of a particular defective item colliding with any of the other $d-1$ defective items is at most $d/B$, which behaves as
\begin{align}
    \frac{d}{B}=\frac{d}{d_{\text{max}}/\beta_n}\leq\frac{d}{d/\beta_n}=\beta_n\rightarrow0, \label{eq:union_bound}.
\end{align}

Secondly, we show that with high probability as $n\rightarrow\infty$, $\widehat{d}$ overestimates $d$. From \eqref{eq:union_bound}, we have
\begin{align}
    \mathbb{E}[\text{\#collisions}]\leq d\beta_n,
\end{align}
where \#collisions refer to the number of items that are in the same bin as any of the other $d-1$ items. By Markov's inequality, we have 
\begin{align}
    \mathbb{P}(\text{\#collisions}\geq d\sqrt{\beta_n})&\leq\sqrt{\beta_n},
\end{align}
which implies the following:
\begin{align}
    \mathbb{P}(d-\text{\#collisions}\geq d-d\sqrt{\beta_n})&\geq1-\sqrt{\beta_n}\\
    \implies\mathbb{P}\Big(\frac{d-\text{\#collisions}}{1-\sqrt{\beta_n}}\geq d\Big)&\geq1-\sqrt{\beta_n}.
\end{align}
Since $(\text{\#positive bins}\geq d-\text{\#collisions})$ always holds, we have $\mathbb{P}(\widehat{d}\geq d) \geq 1-\sqrt{\beta_n}$, which tends to one because $\beta_n\rightarrow0$.
%This is because whenever we want to increase the \#collisions, we must move a single defective item (where it is the only defective item in its bin) to another positive bin. This increases the \#collisions by 1 or 2 but always decrease the \#positive bins by 1.

Finally, we derive the new upper bound for $T$. After estimating $d$, we have used $B=d_{\text{max}}/\beta_n$ number of tests and have a remaining budget of $\gamma-1$ per item. We discard the bins (groups) that returned a negative outcome; instead of continuing with $n$ items, we continue with less than or equal to $(d\times\text{bin size})$ items. To simplify notation, our updated inputs (labeled with subscript ``new'') are
\begin{align}
    n_{\text{new}}&=\frac{\beta_n dn}{d_{\text{max}}},\;d_{\text{new}}=\widehat{d},\;\gamma_{\text{new}}=\gamma-1.
\end{align}
We can then run Algorithm \ref{alg:adaptive_algo} to recover the defective set. Substituting our updated inputs into \eqref{eq:upper_bound} and using $M=\big(\frac{\beta_n dn}{d_{\text{max}}\widehat{d}}\big)^{\frac{\gamma-2}{\gamma-1}}$, we have the following bound for $T$:
\begin{align}
    T&\leq \frac{d_{\text{max}}}{\beta_n}+\frac{\beta_ndn}{d_{\text{max}}(\frac{\beta_ndn}{d_{\text{max}}\widehat{d}})^{\frac{\gamma-2}{\gamma-1}}}+(\gamma-2)d\Big(\frac{\beta_ndn}{d_{\text{max}}\widehat{d}}\Big)^{\frac{1}{\gamma-1}},
\end{align}
which simplifies to
\begin{align}
    T&\leq\frac{d_{\text{max}}}{\beta_n}+(\widehat{d}-2d+\gamma d)\Big(\frac{\beta_ndn}{d_{\text{max}}\widehat{d}}\Big)^{\frac{1}{\gamma-1}}\\
    &\stackrel{(a)}{\leq}\frac{d_\text{max}}{\beta_n}+\Big(\frac{d}{1-\sqrt{\beta_n}}-2d+\gamma d\Big)\Big(\frac{\beta_nn}{d_{\text{max}}}\Big)^{\frac{1}{\gamma-1}}, \label{eq:binning_method}
\end{align}
where we used $d\leq\widehat{d}\leq\frac{d}{1-\sqrt{\beta_n}}$ in (a).

\textit{Comparisons:} By using $T$ satisfying the derived upper bounds, the first approach recovers the defective set with zero error probability, whereas the second approach recovers the defective set with a small error probability determined by the $\beta_n$ parameter.
Referring to \eqref{eq:upper_bound_method} and \eqref{eq:binning_method}, we consider two examples to compare the bounds on $T$. The first example is when $d_{\text{max}}=d$, and the second example is when $\gamma d\ll d_{\text{max}}\ll n$.

For $d_{\text{max}}=d$, as we would naturally expect,  \eqref{eq:upper_bound_method} is the better bound; its leading term is $\gamma d\big(\frac{n}{d}\big)^{1/\gamma}$.  In particular, we note the following two cases: (i) If $\beta_n \ll \frac{1}{ \gamma(\frac{n}{d})^{1/\gamma} }$, then the $\frac{d_{\text{max}}}{\beta_n}$ term in \eqref{eq:binning_method} is strictly higher than $\gamma d\big(\frac{n}{d}\big)^{1/\gamma}$; (ii) If $\beta_n \gg \frac{1}{ \gamma(\frac{n}{d})^{1/\gamma} }$, then some simple algebra gives $\frac{\beta_n n}{d} \gg \frac{1}{\gamma}\big(\frac{n}{d}\big)^{(\gamma-1)/\gamma}$, which implies that the $\gamma d\big(\frac{\beta_n n}{d}\big)^{1/(\gamma-1)}$ term from \eqref{eq:binning_method} is strictly higher than $\gamma d\big(\frac{n}{d}\big)^{1/\gamma}$ (note that $\big(\frac{1}{\gamma}\big)^{1/(\gamma-1)} = \Theta(1)$).

For $\gamma d\ll d_{\text{max}}\ll n$, the choice of $\beta_n$ can impact which bound is smaller.  First note that the dominating term in \eqref{eq:upper_bound_method} is $d_{\text{max}}\big(\frac{n}{d_{\text{max}}}\big)^{1/\gamma}$. Since the dominating term $\max\big\{\frac{d_{\text{max}}}{\beta_n},\gamma d\big(\frac{\beta_nn}{d_{\text{max}}}\big)^{1/(\gamma-1)}\big\}$ in \eqref{eq:binning_method} is not obvious, we consider both possibilities: (i) $d_{\text{max}}\big(\frac{n}{d_{\text{max}}}\big)^{1/\gamma}\gg\frac{d_{\text{max}}}{\beta_n}$ whenever $\beta_n\gg\big(\frac{d_{\text{max}}}{n}\big)^{1/\gamma}$; and (ii) $d_{\text{max}}\big(\frac{n}{d_{\text{max}}}\big)^{1/\gamma}\gg \gamma d\big(\frac{\beta_nn}{d_{\text{max}}}\big)^{\frac{1}{\gamma-1}}$ whenever $\beta_n\ll\big(\frac{d_{\text{max}}}{\gamma d}\big)^{\gamma-1}\big(\frac{d_{\text{max}}}{n}\big)^{1/\gamma}$.
Combining these cases, we see that if $\beta_n$ is in the range $\big(\frac{d_{\text{max}}}{n}\big)^{1/\gamma}\ll\beta_n\ll\big(\frac{d_{\text{max}}}{\gamma d}\big)^{\gamma-1}\big(\frac{d_{\text{max}}}{n}\big)^{1/\gamma}$, the dominating term in \eqref{eq:upper_bound_method} is greater than the dominating term in \eqref{eq:binning_method}. 

Since we have assumed $\beta_n$ to be decaying, we briefly discuss conditions under which the requirement $\big(\frac{d_{\text{max}}}{n}\big)^{1/\gamma}\ll\beta_n$ is consistent with this assumption.  While this lower bound on $\beta_n$ may not always vanish as $n \to \infty$, it does so in broad scaling regimes, including the following: $\gamma\in\Theta((\log n)^c)$ for some $c\in[0,1)$, and $d_{\text{max}}=d=\Theta(n^\theta)$ for some $\theta\in(0,1)$. To see this, note that
\begin{align}
    \lim_{n\rightarrow\infty}\log\Big(\frac{d_{\text{max}}}{n}\Big)^{\frac{1}{\gamma}}
    =\lim_{n\rightarrow\infty}(\alpha-1)(\log n)^{1-c}
    =-\infty, \label{eq:decaying_proof}
\end{align}
and that taking $\exp(\cdot)$ on both sides gives the desired result.

Hence, for $\beta_n$ in the appropriate range, when $d_{\text{max}}$ is close to $d$, using the upper bound directly in Algorithm \ref{alg:adaptive_algo} leads to a smaller $T$. On the other hand, when $\gamma d\ll d_{\text{max}}\ll n$, using the binning method before Algorithm \ref{alg:adaptive_algo} leads to a smaller $T$.

\subsection{Algorithm for $\rho$-Sized Tests} \label{sec:adap_algo_rho}

\begin{algorithm}[t]
    \begin{algorithmic}[1]
        \REQUIRE Population of items, number of items $n$, number of defective items $d$, and test size restriction $\rho$
        \STATE Initialize defective set $\mathcal{D}\leftarrow\emptyset$
        \STATE Randomly group $n$ items into $n/\rho$ groups of size $\rho$
        \FOR{each group $G_i$ where $i\in\{1,2,\dots,n/\rho\}$}{
        \WHILE{testing $G_i$ returns a positive outcome}{
        \STATE run Algorithm 4 on $G_i$ and add its one defective item output $d^*$ into $\mathcal{D}$
        \STATE $G_i\leftarrow G_i\setminus\{d^*\}$
        }\ENDWHILE
        }\ENDFOR
        \RETURN $\mathcal{D}$
    \end{algorithmic}
    \caption{Adaptive algorithm for $\rho$-sized tests \label{alg:adaptive_algo_rho}}
\end{algorithm}

\begin{algorithm}[t]
    \begin{algorithmic}[1]
        \REQUIRE a group of items $G_i$
        \STATE If $G_i$ consists of a single item, return that item.
        \STATE Pick half of the items in $G_i$ and call this set $G_i'$. Perform a single test on $G_i'$.
        \STATE If the test is positive, set $G_i\leftarrow G_i'$. Otherwise, set $G_i\leftarrow G_i\setminus G_i'$. Return to step 1.
        \RETURN $\mathcal{D}$
    \end{algorithmic}
    \caption{Binary splitting sub-routine \label{alg:binary_splitting}}
\end{algorithm}

While our main focus is on the $\gamma$-divisible constraint (motivated by it having larger gaps in the bounds \cite{Ven19}), here we briefly pause to provide a simple adaptive algorithm for the $\rho$-sized test constraint, shown in Algorithm \ref{alg:adaptive_algo_rho}. This is a direct modification of Hwang's generalized binary splitting algorithm \cite{Hwa72}, in which we divide the $n$ items into $\frac{n}{\rho}$ groups of size $\rho$, instead of $d$ groups of size $\frac{n}{d}$ as in the original algorithm.

\textit{Analysis:} Let $d_1,\dotsc,d_{\frac{n}{\rho}}$ be the number of defective items in each of the initial $\frac{n}{\rho}$ groups. Note that the assumption $\rho\in o(\frac{n}{d})$ (see Section \ref{sec:previous}) implies $d\in o(\frac{n}{\rho})$, most groups will not have a defective item.  In the binary splitting stage of the algorithm, we can round the halves in either direction if they are not an integer. Hence, for each of the initial $\frac{n}{\rho}$ groups, we take at most $\lceil\log_2\rho\rceil$ adaptive tests to find a defective item, or one test to confirm that there are no defective item. Therefore, for each of the initial $\frac{n}{\rho}$ groups, we need $\max\{1,d_i\log_2\rho+O(d_i)\}$ tests to find $d_i$ defective items. Summing across all $\frac{n}{\rho}$ groups, we need a total of $T=\sum_{i=1}^{n/\rho}\max\{1,d_i\log_2\rho+O(d_i)\}$ tests. This leads to the following upper bound:
\begin{align}
    T&\leq\frac{n}{\rho}+d\log_2\rho+O(d)\\
    &\stackrel{(a)}{=}\frac{n}{\rho}(1+o(1))+d\log_2\rho,
\end{align}
where (a) uses $d\in o\big(\frac{n}{\rho}\big)$. With the further condition $\rho\in O\big(\frac{n}{d\log(n/d)}\big)$, we have $\frac{n}{\rho}\in\Omega\big(d\log\big(\frac{n}{d}\big)\big)$ and $d\log\rho\in o\big(d\log\big(\frac{n}{d}\big)\big)$. Thus, we can further simplify to get
\begin{align}
    T\leq\frac{n}{\rho}(1+o(1)).
\end{align}
This upper bound is tight in the sense that attaining vanishing error probability trivially requires a fraction $1-o(1)$ of the items to be tested at least once, which implies $T\geq\frac{n}{\rho}(1 - o(1))$ by the $\rho$-sized test constraint.

\subsection{Proof of Theorem \ref{thm:counting_bound} (Counting-Based Bound)} \label{sec:counting_bound_proof}
Given a population of $n$ objects, we write $\Sigma_{n,d}$ for the collection of subsets of size $d$ from the population. Furthermore, we write $\mathcal{D}$ for the true defective set.

We follow the steps of \cite{Bal13} as follows: The testing procedure defines a mapping $\theta:\Sigma_{n,d}\rightarrow\{0,1\}^T$. Given a putative defective set $S\in\Sigma_{n,d}$, $\theta(S)$ is the vector of test outcomes, with positive tests represented as 1s and negative tests represented as 0s. For each $\textbf{y}\in\{0,1\}^T$, we write $\mathcal{A}_{\textbf{y}}\subseteq\Sigma_{n,d}$ for the inverse image of \textbf{y} under $\theta$:
\begin{align}
    \mathcal{A}_{\textbf{y}}=\theta^{-1}(\textbf{y})=\{S\in\Sigma_{n,d}:\theta(S)=\textbf{y}\}.
\end{align}

The role of an algorithm that decodes the outcome of the tests is to mimic the effect of the inverse image map $\theta^{-1}$. Given a test output \textbf{y}, the optimal decoding algorithm would use a lookup table to find the inverse image $\mathcal{A}_{\textbf{y}}$. If this inverse image $\mathcal{A}_{\textbf{y}}=\{S\}$ has size $|\mathcal{A}_{\textbf{y}}|=1$, we can be certain that the defective set was $S$. In general, if $|\mathcal{A}_{\textbf{y}}|\geq1$, we cannot do better than pick uniformly among $\mathcal{A}_{\textbf{y}}$, with success probability $\frac{1}{|\mathcal{A}_{\textbf{y}}|}$ (We can ignore empty $\mathcal{A}_{\textbf{y}}$, since we are only concerned with vectors \textbf{y} that occur as a test output).

Hence, overall, the probability of recovering a defective set $S$ is $\frac{1}{|\mathcal{A}_{\theta(S)}|}$, depending only on $\theta(S)$. We can write the following expression for the success probability, conditioning over all the equiprobable values of the defective set:
\begin{align}
    \mathbb{P}(\text{suc})&\stackrel{(a)}{=}\sum_{S\in\Sigma_{n,d}}\mathbb{P}(\text{suc}|\mathcal{D}=S)\frac{1}{{n\choose d}}\\
    &=\frac{1}{{n\choose d}}\sum_{S\in\Sigma_{n,d}}\sum_{\textbf{y}\in\{0,1\}^T}\mathds{1}(\theta(S)=\textbf{y})\mathbb{P}(\text{suc}|\mathcal{D}=S)\\
    &=\frac{1}{{n\choose d}}\sum_{S\in\Sigma_{n,d}}\sum_{\textbf{y}\in\{0,1\}^T: |\mathcal{A}_{\textbf{y}}|\geq1}\mathds{1}(\theta(S)=\textbf{y})\frac{1}{|\mathcal{A}_{\textbf{y}}|}\\
    &=\frac{1}{{n\choose d}}\sum_{\textbf{y}\in\{0,1\}^T \,:\, |\mathcal{A}_{\textbf{y}}|\geq1}\frac{1}{|\mathcal{A}_{\textbf{y}}|}\Big(\sum_{S\in\Sigma_{n,d}}\mathds{1}(\theta(S)=\textbf{y})\Big)\\
    &=\frac{1}{{n\choose d}}\sum_{\textbf{y}\in\{0,1\}^T \,:\, |\mathcal{A}_{\textbf{y}}|\geq1}\frac{1}{|\mathcal{A}_{\textbf{y}}|}|\mathcal{A}_{\textbf{y}}|\\
    &=\frac{|\{\textbf{y}\in\{0,1\}^T \,:\, |\mathcal{A}_{\textbf{y}}|\geq1\}|}{{n\choose d}}\\
    &\stackrel{(b)}{\leq}\frac{|\{\text{\textbf{y} with $\leq\gamma d$ ones}\}|}{{n\choose d}}
    =\frac{\sum_{i=0}^{\gamma d}{T\choose i}}{{n\choose d}},
\end{align}
where (a) uses the law of total probability and the uniform prior on $\mathcal{D}$, and (b) uses the fact that at most $\gamma d$ test outcomes can be positive, even in the adaptive setting. This is because adding another defective always introduces at most $\gamma$ additional positive tests.

\subsection{Proof of Theorem \ref{thm:gamma_lower_bound_adap} (General Converse for $\gamma$-Divisible Items)} \label{sec:gamma_lowerbound_adap_proof}

From the counting bound in (2), we upper bound the sum of binomial coefficients \cite[Section~4.7.]{Ash90} to obtain
\begin{align}
    \mathbb{P}(\text{suc})&\leq\frac{e^{TH_2(\frac{\gamma d}{T})}}{{n\choose d}}\equiv\delta, \label{eq:weak_counting_bound}
\end{align}
where $H_2(\cdot)$ is the binary entropy function in nats. From \eqref{eq:weak_counting_bound}, we have $e^{TH_2(\frac{\gamma d}{T})}/{n\choose d}=\delta$, which implies that
\begin{align}
    \log\bigg(\delta{n\choose d}\bigg)&=TH_2\Big(\frac{\gamma d}{T}\Big)\\
    &=\gamma d\log{\frac{T}{\gamma d}}+(T-\gamma d)\log\frac{1}{1-\frac{\gamma d}{T}}\\
    &\stackrel{(a)}{=}\gamma d\log\frac{T}{\gamma d}+\gamma d(1+o(1)), \label{eq:simplified_entropy}
\end{align}
where (a) uses a Taylor expansion and the fact that $\frac{\gamma d}{T}\in o(1)$ from \eqref{eq:T_gamma_d_ineq}; hence, we have $(1-\frac{\gamma d}{T})^{-1}=\exp(\frac{\gamma d}{T})(1+o(1))$ which is used to obtain the simplification. Rearranging \eqref{eq:simplified_entropy}, we obtain
\begin{align}
    \gamma d\log\frac{T}{\gamma d}&=\log\bigg(\delta{n\choose d}\bigg)-\gamma d(1+o(1))\\
    \implies \log\frac{T}{\gamma d}&=\frac{1}{\gamma d}\log\bigg(\delta{n\choose d}\bigg)-(1+o(1)),
\end{align}
which gives
\begin{align}
    T&=e^{-(1+o(1))}\gamma d\bigg(\delta{n\choose d}\bigg)^{\frac{1}{\gamma d}}\\
    &\stackrel{(a)}{\geq}e^{-(1+o(1))}\gamma d\delta^{\frac{1}{\gamma d}}\Big(\frac{n}{d}\Big)^{\frac{1}{\gamma}}, \label{eq:tests}
\end{align}
where (a) follows from the fact that ${n\choose d}\geq\big(\frac{n}{d}\big)^d$.

The proof is completed by noting that for a fixed target success probability $\delta=1-\epsilon$, $\delta^{1/(\gamma d)}\rightarrow1$ as $\gamma d\rightarrow\infty$.

\subsection{Proof of Theorem \ref{thm:gamma_upperbound_adap} (Adaptive Algorithm Performance)} \label{sec:gamma_upperbound_adap_proof}

Similar to Hwang's generalized binary splitting algorithm \cite{Hwa72}, the idea behind the parameter $M$ in Algorithm \ref{alg:adaptive_algo} is that when $d$ becomes large, having large groups during the initial splitting stage is wasteful, as it results in each test having a very high probability of being positive (not very informative). Hence, we want to find the appropriate group sizes that result in more informative tests to minimize the number of tests. 

Each stage (outermost for-loop in Algorithm \ref{alg:adaptive_algo}) here refers to the process where all groups of the same sizes are split into smaller groups (as seen in Figure \ref{fig:adaptive_algo}). We let $M$ be the group size at the initial splitting stage of the algorithm. The algorithm first tests $n/M$ groups of size $M$ each,\footnote{Note that $\frac{n}{M}$ is an integer for our chosen $M$ following \eqref{eq:upper_bound}, which gives $\frac{n}{M}=d(\frac{n}{d})^{1/\gamma}$, and $(\frac{n}{d})^{1/\gamma}$ was assumed to be an integer earlier.} then steadily decrease the sizes of each group down the stages: $M\rightarrow M^{1-1/(\gamma-1)}\rightarrow M^{1-2/(\gamma-1)}\rightarrow\dots\rightarrow 1$ (see Figure \ref{fig:adaptive_algo} for visualization). Hence, we have $n/M$ groups in the initial splitting and $M^{\frac{1}{\gamma-1}}$ groups in all subsequent splits. 

With the above observations, we can derive an upper bound on the total number of tests needed. We have $n/M$ tests in the first stage. Since we have $d$ defectives and split into $M^{\frac{1}{\gamma-1}}$ sub-groups in subsequent stages, the number of smaller groups that each stage can produce is at most $dM^{\frac{1}{\gamma-1}}$. This implies that the number of tests conducted at each stage is at most $dM^{\frac{1}{\gamma-1}}$, giving the following bound on $T$:
\begin{align}
    T &\leq \frac{n}{M}+(\gamma-1)dM^{\frac{1}{\gamma-1}} \label{eq:upper_bound}.
\end{align}
We optimize with respect to $M$ by differentiating the upper bound and setting it to zero, which gives $M=(\frac{n}{d})^{\frac{\gamma-1}{\gamma}}$. Substituting $M=(\frac{n}{d})^{\frac{\gamma-1}{\gamma}}$ into the general upper bound in \eqref{eq:upper_bound}, we obtain the following upper bound:
\begin{align}
    T\leq\frac{n}{(\frac{n}{d})^{\frac{\gamma-1}{\gamma}}}+(\gamma-1)d\bigg[\Big(\frac{n}{d}\Big)^{\frac{\gamma-1}{\gamma}}\bigg]^{\frac{1}{\gamma-1}}
    % &=d^{\frac{\gamma-1}{\gamma}}n^{1-\frac{\gamma-1}{\gamma}}+d(\gamma-1)\Big(\frac{n}{d}\Big)^{\frac{1}{\gamma}}\\
    % &=d^{1-\frac{1}{\gamma}}n^{\frac{1}{\gamma}}+d(\gamma-1)\Big(\frac{n}{d}\Big)^{\frac{1}{\gamma}}\\
    % &=d\Big(\frac{n}{d}\Big)^{\frac{1}{\gamma}}+d(\gamma-1)\Big(\frac{n}{d}\Big)^{\frac{1}{\gamma}}\\
    =\gamma d\Big(\frac{n}{d}\Big)^{\frac{1}{\gamma}}. \label{eq:adaptive_upper_bound}
\end{align}

\section{The Non-Adaptive Setting} \label{ch:sparse_nonadap_group_testing}

In this section, we develop bounds and algorithms for the non-adaptive setting with $\gamma$-divisible items.

\subsection{Converse Bound for the Near-Constant Tests-Per-Item Design} \label{sec:non_adap_info_theo_lower_bound}

In this section, we present an information-theoretic lower bound on the number of tests for the near-constant test-per-item random design with parameter $\gamma$.  Note that this is in contrast to Theorems \ref{thm:gamma_lowerbound} and \ref{thm:gamma_lower_bound_adap}, which hold for {\em arbitrary} non-adaptive test designs.  Of course, lower bounds for arbitrary designs are generally preferable; however, the converse specific to the random design will be seen to be significantly tighter in denser scaling regimes.  See \cite{Ald14a,Joh16,Li19} for similar design-specific converse bounds in other contexts.

We follow the high-level approach of \cite{Ald15}, showing that if both the Combinatorial Orthogonal Matching Pursuit (COMP) algorithm \cite{Cha14} and the Smallest Satisfying Set (SSS) algorithm \cite{Ald14a} fail, then so does any algorithm.  We proceed by introducing these algorithms formally.

\begin{definition}
    The COMP algorithm for noiseless non-adaptive group testing is given as follows: Mark each item that appears in a negative test as non-defective, and refer to every other item as a possibly defective. We write $\mathcal{PD}$ for the set of such items, yielding $\widehat{\mathcal{D}}_{\text{COMP}} = \mathcal{PD}$.
\end{definition}
We observe that the COMP algorithm succeeds if and only if every non-defective item is included in at least one negative test.

For the SSS algorithm, we first state a key definition, and then describe the algorithm.
\begin{definition}
    We say that a putative defective set $\mathcal{J}$ is \textit{satisfying} if:
    \begin{enumerate}
        \item No negative test contains a member of $\mathcal{J}$.
        \item Every positive test contains at least one member of $\mathcal{J}$.
    \end{enumerate}
\end{definition}
\begin{definition}
    The SSS algorithm for noiseless non-adaptive group testing is given as follows: Find the smallest satisfying set (breaking ties arbitrarily), and take that as the estimate $\widehat{\mathcal{D}}_{\text{SSS}}$.  
\end{definition}
Note that the true defective set $\mathcal{D}$ is certainly a satisfying set, and hence SSS is guaranteed to return a set of no larger size, giving $|\widehat{\mathcal{D}}_{\text{SSS}}| \leq d$. In addition, we can identify a particular failure event for SSS \cite{Ald14a}: If a defective item $i\in\mathcal{D}$ is not the unique defective item in any positive test, then $\mathcal{D}\setminus\{i\}$ will be a smaller satisfying set than $\mathcal{D}$, so SSS is certain to fail.

Following the above outline, our algorithm-independent converse for the near-constant tests-per-item design will be a simple corollary to the following theorem.

\begin{theorem} \label{thm:gamma_lowerbound_ours_nearconstantwtdesign}
    {\em (Design-Specific Converse for COMP and SSS)}
    Under the near-constant tests-per-item design, with $d\in\Theta(n^\theta)$ for some positive constant $\theta\in(0,1)$, and $\gamma\in\Theta\big((\log n)^c\big)$ tests per item for some $c \in [0,1)$, if
    \begin{equation}
        T = \gamma d^{1 + 1/\gamma} (1-\zeta) \label{eq:T_choice}
    \end{equation}
    for fixed $\zeta \in (0,1)$, then we have
    \begin{gather}
          \mathbb{P}( \widehat{\mathcal{D}}_{\text{COMP}} \ne \mathcal{D} ) = 1-o(1) \\
          \mathbb{P}( |\widehat{\mathcal{D}}_{\text{SSS}}| < d ) \ge \frac{1 - o(1)}{(1 - \zeta + o(1))^{\gamma} + 1}. \label{eq:prob_SSS}
    \end{gather}
\end{theorem}
\begin{proof}
    See Section \ref{sec:gamma_lowerbound_ours_proof_SSS}.
\end{proof}

In the scaling regime $\gamma\in\Theta\big((\log n)^c\big)$ for some $c\in[0,1)$, the right-hand side of \eqref{eq:prob_SSS} approaches one if $c>0$ (large $\gamma$), is close to one if $c=0$ (constant $\gamma$) as long as $\gamma$ is large compared to $\frac{1}{\zeta}$, and is always at least $\frac{1}{2} + o(1)$.

\begin{corollary} \label{cor:gamma_lowerbound_ours_nearconstantwtdesign} 
    {\em (Design-Specific Converse for Arbitrary Algorithms)}
    Under the near-constant tests-per-item design, with $d\in\Theta(n^\theta)$ for some positive constant $\theta\in(0,1)$, and $\gamma\in\Theta\big((\log n)^c\big)$ tests per item for some $c \in [0,1)$, if
    \begin{equation}
        T \le \gamma d^{1 + 1/\gamma} (1-\zeta) \label{eq:T_choice2}
    \end{equation}
    for some $\zeta > 0$, then the error probability is bounded away from zero regardless of the decoding algorithm.
\end{corollary}
\begin{proof}
    It was proved in \cite{Ald15} that if $\mathbb{P}( \widehat{\mathcal{D}}_{\text{COMP}} \ne \mathcal{D} ) + \mathbb{P}( |\widehat{\mathcal{D}}_{\text{SSS}}| < d ) \ge 1 + \epsilon$ for some $\epsilon>0$, then the error probability is at least $\frac{\epsilon}{2}$ for an arbitrary algorithm.  Hence, the desired result follows immediately from Theorem \ref{thm:gamma_lowerbound_ours_nearconstantwtdesign}; it suffices to consider \eqref{eq:T_choice2} holding with equality, because any decoding algorithm can always choose to ignore some of the tests.
\end{proof}

We observe that the converse in Corollary \ref{cor:gamma_lowerbound_ours_nearconstantwtdesign} is tighter (i.e., has a higher lower bound on $T$) than that of Theorem \ref{thm:gamma_lowerbound} when $d$ is ``large'', i.e., when $\theta$ is above $\frac{1}{2}$, or in particular, close to one.

\subsection{Analysis of the DD Algorithm with $\gamma$-Divisible Items} \label{sec:non_adap_algo_gamma}

We continue focus on the random near-constant tests-per-item design for the $\gamma$-divisible items constraint, where $\gamma\in o\big(\log\big(\frac{n}{d}\big)\big)$ tests are chosen uniformly at random with replacement for each item.   However, we now turn our attention to upper bounds.
% More specifically, $\gamma$ entries of each column in $\mathsf{X}$ are selected uniformly at random with replacement and set to one. The remaining entries are set to zero. Some items may be in fewer than $\gamma$ tests, hence the terminology ``near-constant".

We will use the Definite Defectives (DD) decoding algorithm \cite{Ald14a}, which is defined as follows.
\begin{definition}
    The Definite Defectives (DD) algorithm for noiseless non-adaptive group testing has two keys steps.
    \begin{enumerate}
        \item Since $y_t=1$ if and only if the test pool contains a defective item, we can be sure that each item that appears in a negative test is not defective. We form a list of such items from all the negative tests, which we refer to as the guaranteed non-defective ($\mathcal{ND}$) set. The rest of the items, $ \mathcal{PD} :=\{1,\dots,n\}\setminus\mathcal{ND}$, form the possibly defective ($\mathcal{PD}$) set.
        \item Since every positive test must contain at least one defective item, if a test with $Y=1$ contains exactly one item from  $\mathcal{PD}$, then we can be certain that the item in question is defective. The DD algorithm estimates $\mathcal{D}$ using $\widehat{\mathcal{D}}$ to be the set of $\mathcal{PD}$ items which appear in a positive test with no other $\mathcal{PD}$ item.
    \end{enumerate}
\end{definition}

Note that the first step is the same as COMP; it never makes a mistake in adding to $\mathcal{ND}$ (items are correctly marked as non-defective).  Similarly, the second step never makes a mistake in adding to $\widehat{\mathcal{D}}$ (items are correctly marked as defective). Hence, any errors due to DD come from marking a true defective as non-defective in the second step, meaning that $\widehat{\mathcal{D}}\subseteq\mathcal{D}$. The choice to mark all remaining items as non-defective is motivated by the sparsity of the problem (recall that $d\in o(n)$), since \textit{a priori} an item is much less likely to be defective than non-defective. 

By analyzing the DD algorithm, we obtain the following theorem.
\begin{theorem} \label{thm:gamma_upperbound_nonadap}
    For $\gamma\in\Theta\big((\log n)^c\big)$ for some $c\in[0,1)$, $d\in\Theta(n^\theta)$ for some $\theta\in(0,1)$, $\alpha_2\in(0,1)$, and any function $\beta_n$ decaying as $n$ increases, under the near-constant tests-per-item design with parameter $\gamma$ and a number of tests given by 
    \begin{align}
        T=\gamma d\max
        \bigg\{e^{\frac{1}{\alpha_2}H_2(\max\{\alpha_2,\frac{1}{2}\})}\Big(\frac{d}{\beta_n}\Big)^{\frac{1}{\alpha_2\gamma}},
        2^{1/\gamma}\Big(\frac{n-d}{d}\Big)^{\frac{1}{\gamma}}\Big(\frac{d}{\beta_n}\Big)^{\frac{1}{(1-\alpha_2)\gamma^2}}\bigg\}, \label{eq:T_long}
    \end{align}
    the DD algorithm ensures an error probability of at most
    \begin{align}
        P_e \le \exp\bigg(-\frac{3d}{16}\Big(\frac{\beta_n}{d}\Big)^{\frac{1}{(1-\alpha_2)\gamma)}}\bigg)+2\exp(-2(\gamma d)^{1/3})+2\beta_n(1+o(1)).
    \end{align}
\end{theorem}
\begin{proof}
    See Section \ref{sec:gamma_upperbound_nonadap}.
\end{proof}

In order to better understand this bound on $T$, we simplify it in two different scaling regimes:
\begin{enumerate}
    \item \underline{Large $\gamma$:} $\gamma\in\Theta((\log n)^c)$ for some $c\in(0,1)$, and $d\in\Theta(n^\theta)$ for some $\theta\in(0,1)$
    \item \underline{Constant $\gamma$:} $\gamma\in O(1)$, and $d\in\Theta(n^\theta)$ for some $\theta\in(0,1)$.
\end{enumerate}
In both regimes, we assume that $\beta_n$ is a slowly decaying term, since we are primarily interested in attaining $P_e \to 0$ rather than the speed of convergence.  

It will be useful to compare the bounds in terms of the following quantity:
\begin{align}
    \eta&=\lim_{n\rightarrow\infty}\frac{\log(\frac{n}{d})}{\gamma\log(\frac{T}{\gamma d})}. \label{eq:eta}
\end{align}
Observe that for any fixed value of $\eta > 0$, re-arranging gives $T = \gamma d \big( \big( \frac{n}{d} \big)^{\frac{1}{\gamma}} \big)^{\frac{1+o(1)}{\eta}}$.  We henceforth use the notation $\tilde{O}(\cdot)$ and $\tilde{\Omega}(\cdot)$ the denote the asymptotic behavior of $T$ up to factors that do not impact $\eta$, and accordingly omit $\beta_n$ from such expressions.

For regime 1 (large $\gamma$), letting $\alpha_2$ be a fixed constant close to one, we obtain that $e^{\frac{1}{\alpha_2}H_2(\max\{\alpha_2,\frac{1}{2}\})}$ can be made arbitrarily close to one, and in addition, the assumed scaling on $\gamma$ and $d$ gives
% in addition, we claim that $\big(\frac{n-d}{d}\big)^{1/\gamma}d^{\frac{1}{(1-\alpha_2)\gamma^2}}=\big(\frac{n}{d}\big)^{(1-o(1))/\gamma}$. To see the latter, note that
\begin{align}
    \Big(\frac{n-d}{d}\Big)^{\frac{1}{\gamma}}d^{\frac{1}{(1-\alpha_2)\gamma^2}}
    &=\Big(\frac{n-d}{d^{1-O(1/\gamma)}}\Big)^{\frac{1}{\gamma}}\\
    &=\bigg(\frac{n^{\frac{1}{1-O(1/\gamma)}}(1-O(d/n))^{\frac{1}{1-O(1/\gamma)}}}{d}\bigg)^{\frac{1-O(1/\gamma)}{\gamma}}\\
    &=\Big(\frac{n}{d}\Big)^{\frac{1-o(1)}{\gamma}}.
\end{align}
By substituting into \eqref{eq:T_long} and omitting $\beta_n$ as explained above, we obtain
\begin{align}
    T=\tilde{O}\Big(\gamma d\max\{n^\theta,n^{1-\theta}\}^{\frac{1}{\gamma}}\Big),
\end{align}
which matches the $\Omega\big(\gamma d\max\{n^\theta,n^{1-\theta}\}^{\frac{1}{\gamma}}\big)$ lower bound obtained by combining Theorem \ref{thm:gamma_lower_bound_adap} and Corollary \ref{cor:gamma_lowerbound_ours_nearconstantwtdesign}.

We plot $\eta$ against $\theta\in(0,1)$ in Figure \ref{fig:plot_gamma_large_case} to show how the asymptotic bound of the DD algorithm compares to the converse (Theorems \ref{thm:gamma_lowerbound} and \ref{thm:gamma_lower_bound_adap}) and the COMP bound (Theorem \ref{thm:gamma_upperbound}).
\begin{figure}[t]
    \centering
    \includegraphics[scale=0.8]{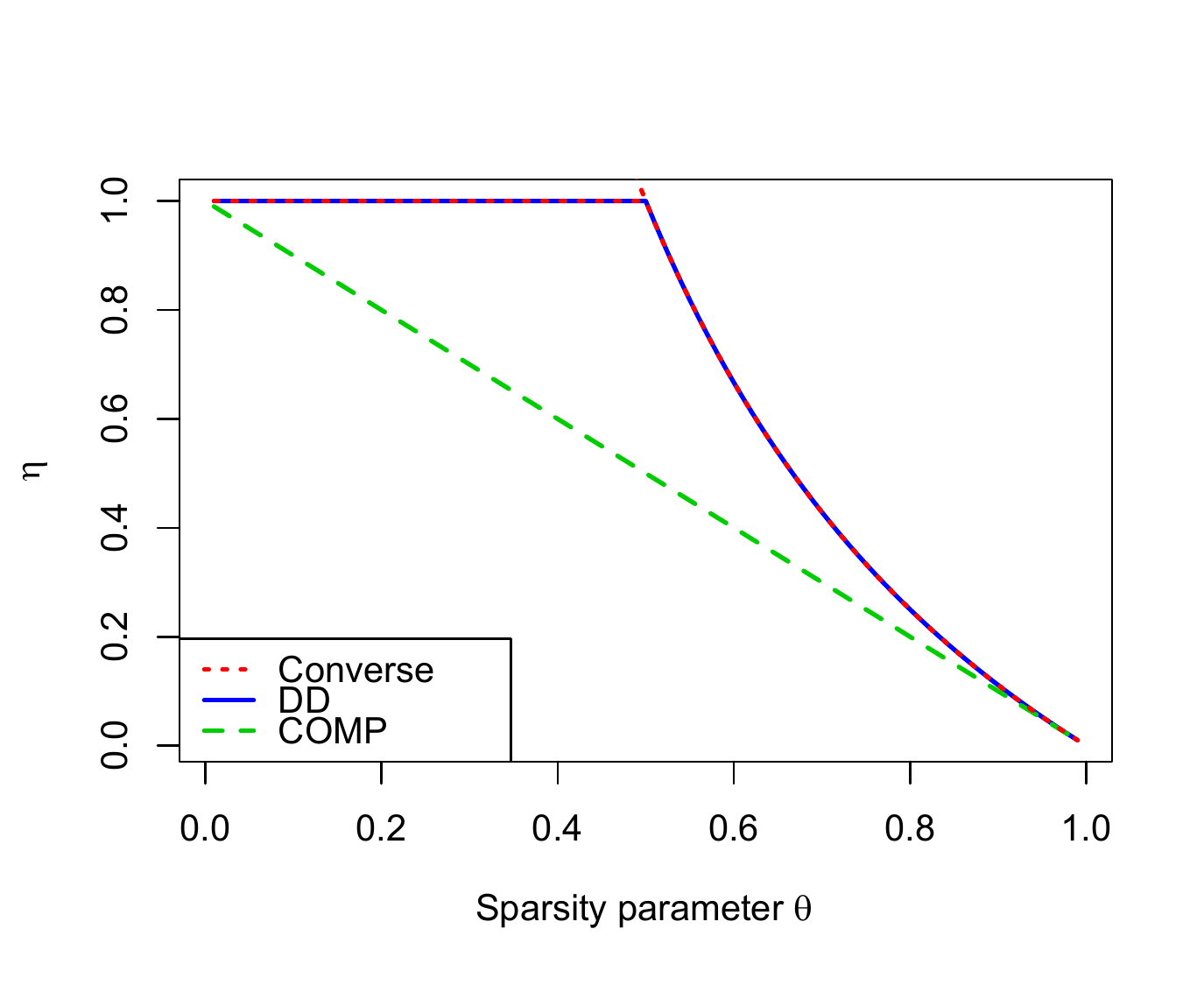}
    \caption{Plots of the variable $\eta$ in \eqref{eq:eta} against the sparsity parameter $\theta$ for the converse, DD algorithm, and COMP algorithm, when $n\rightarrow\infty$ and $\gamma=(\log n)^c$ for some $c\in(0,1)$.} \label{fig:plot_gamma_large_case}
\end{figure}
Note that for the COMP algorithm, we have omitted $\epsilon$ in our asymptotic bound (similarly to $\beta_n$ above), giving $T=\tilde{O}(\gamma dn^{1/\gamma})$. From Figure \ref{fig:plot_gamma_large_case}, we see that the DD algorithm performs better than the COMP algorithm, and achieves the optimal value of $\eta$ for all $\theta\in(0,1)$.

For regime 2 (constant $\gamma$), we similarly substitute the scaling laws into \eqref{eq:T_long} (and omit $\beta_n$) to get
\begin{align}
    T=\tilde{O}\Big(\gamma d\max\Big\{n^{\frac{\theta}{\alpha_2\gamma}},n^{\frac{1-\theta}{\gamma}+\frac{\theta}{(1-\alpha_2)\gamma^2}}\Big\}\Big). 
\end{align}
We numerically optimize with respect to $\alpha_2$ to obtain our bound on $T$. Figure \ref{fig:plot_gamma_10_case} shows how the asymptotic bound of the DD algorithm compares to the converse (Theorems \ref{thm:gamma_lowerbound} and \ref{thm:gamma_lower_bound_adap}) and the COMP bound (Theorem \ref{thm:gamma_upperbound}), when $\gamma=10$.  We see that DD again significantly outperforms COMP, but falls short of the converse.

\begin{figure}[t]
    \centering
    \includegraphics[scale=0.8]{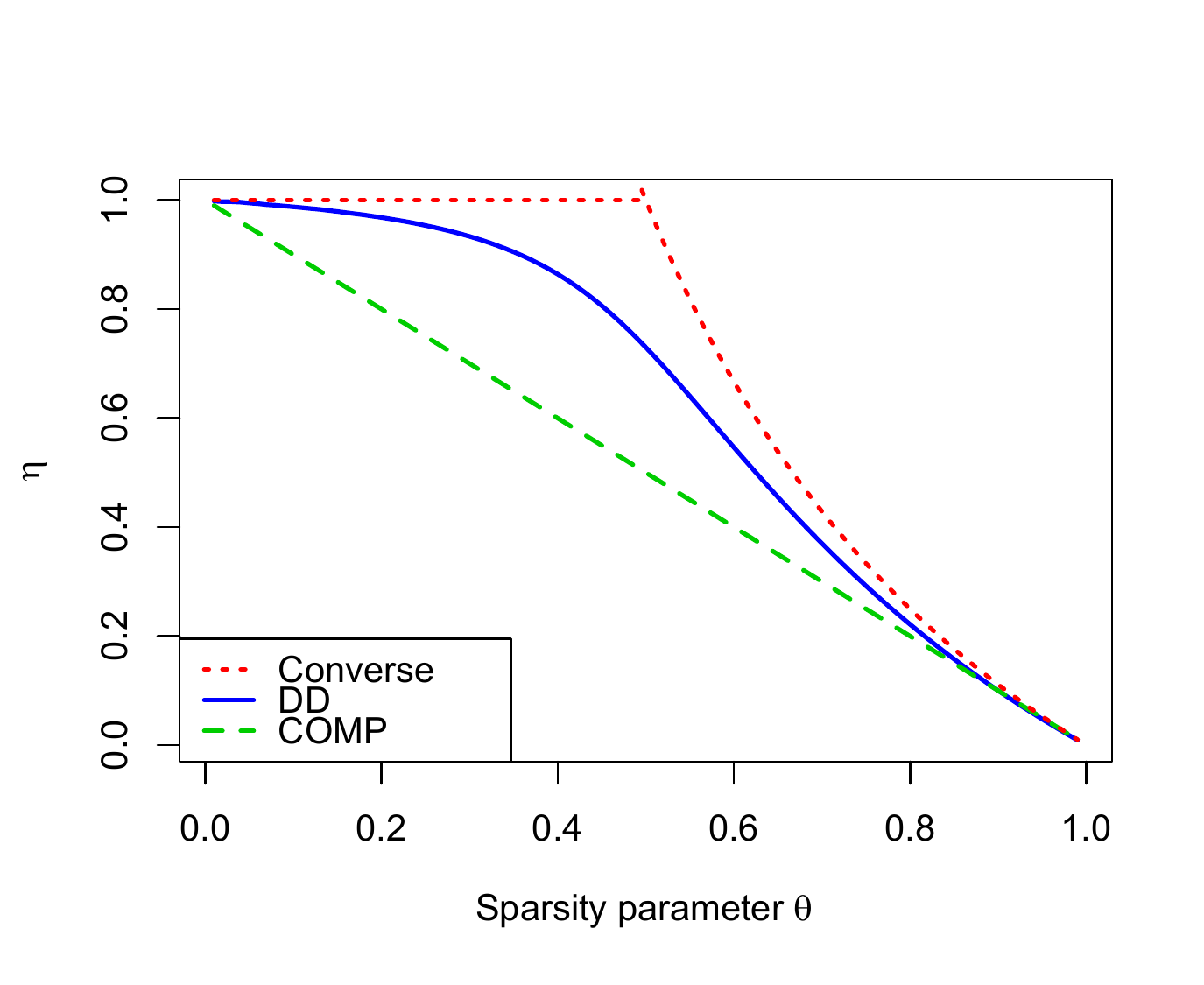}
    \caption{Plots of the variable $\eta$ in \eqref{eq:eta} against the sparsity parameter $\theta$ for the converse, DD algorithm, and COMP algorithm, when $n\rightarrow\infty$ and $\gamma=10$.} \label{fig:plot_gamma_10_case}
\end{figure}

This last example provides a useful point of comparison with the concurrent work of \cite{Geb20}.  It was shown therein that the converse curve in Figure \ref{fig:plot_gamma_10_case} can in fact be matched exactly.  Thus, the proof techniques of \cite{Geb20} appear to be more powerful than ours in the regime $\gamma = \Theta(1)$.  On the other hand, regimes where $\gamma \to \infty$, such as that considered in Figure \ref{fig:plot_gamma_large_case}, are not considered in \cite{Geb20}.

\subsection{Preliminary Definitions and Results} \label{sec:prelim_def_and_results}

Before presenting the main proofs, we introduce some useful definitions and auxiliary results.

\begin{definition} \label{def:mask}
    Consider an item $i$ and a set of items $\mathcal{L}$ not including $i$. We say that item $i$ is \textit{masked} by $\mathcal{L}$ if every test that includes $i$, also includes at least one member of $\mathcal{L}$.
\end{definition}

\begin{definition} \label{def:collision}
    The number of \textit{collisions} between a given item $i$ and a given set of items $\mathcal{L}$ refers to the number of tests selected by the near-constant tests-per-item design for item $i$ (including repetitions in the sampling with replacement) that also include at least one member of $\mathcal{L}$.
\end{definition}

Next, we introduce some auxiliary lemmas that will be used throughout our main proofs. 

\begin{lemma} \label{lem:(1-o(1))^gamma->1}
    If $\gamma\in\Theta\big((\log n)^c\big)$ for some $c\in [0,1)$, and $d\in\Theta(n^{\theta})$ for some $\theta\in(0,1]$, then we have $\big(1\pm\frac{1}{d^{\kappa / \gamma}}\big)^{\gamma}=1\pm o(1)$ for any fixed $\kappa > 0$. 
\end{lemma}
\begin{proof}
    Throughout the proof, we write $g_n \gg f_n$ as a shorthand for $f_n = o(g_n)$.
    Since $\big(1\pm\frac{1}{d^{\kappa/\gamma}}\big)^{\gamma}\rightarrow1$ if $d^{\kappa/\gamma}\gg\gamma$, it suffices to show that $d^{\kappa/\gamma}\gg\gamma$. We have
    \begin{align}
        &\kappa\theta(\log n)^{1-c}\gg c\log\log n\\
        \implies&\frac{\kappa\theta}{(\log n)^c}(\log n)\gg\log(\log n)^c\\
        \implies&\log(n^{\kappa\theta/(\log n)^c})\gg\log(\log n)^c\\
        \implies&\quad n^{\kappa\theta/(\log n)^c}\gg(\log n)^c.
    \end{align}
    Since $d\in\Theta(n^{\theta})$ and $\gamma\in\Theta\big((\log n)^c\big)$, by substitution in the above equation, we get $d^{1/\gamma}\gg\gamma$ which completes the proof.
\end{proof}

Let $W^{(\mathcal{D})}$ be the total number of positive tests containing at least one item from $\mathcal{D}$. To understand the distribution of this quantity, it is helpful to think of the process by which elements of the columns are sampled as a \textit{coupon collector} problem, where each coupon corresponds to one of the $T$ tests. For a single defective item, $W^{(\{i\})}$ is the number of distinct coupons selected when $\gamma$ coupons are chosen uniformly at random from a population of $T$ coupons. In general, for the defective set $\mathcal{D}$ of size $d$, the independence of distinct columns means that $W^{(\mathcal{D})}$ is the number of distinct coupons collected when choosing $\gamma d$ coupons uniformly at random from a population of $T$ coupons. We now give a concentration measure result for $W^{(\mathcal{D})}$ around its mean, which follows via the same arguments as the unconstrained setting \cite{Joh16}.

\begin{lemma} \label{lem:concentration_pos_tests}
    When making $\gamma d\in o(T)$ draws with replacement from a total of $T$ coupons, the total number of distinct coupons $W^{(\mathcal{D})}$ satisfies
    \begin{align}
        \mathbb{P}\big( |W^{(\mathcal{D})}-\gamma d(1-\delta_n)|\geq(\gamma d)^{2/3}\big) \leq 2\exp(-2(\gamma d)^{1/3}),
    \end{align}
    where $\delta_n\in O\big(\frac{\gamma d}{T}\big)$.
\end{lemma}
\begin{proof}
    For any coupon, the probability of not being selected is $1-\big(1-\frac{1}{T}\big)^{\gamma d}$, yielding
    \begin{align}
        \E[W^{(\mathcal{D})}]&=\bigg(1-\Big(1-\frac{1}{T}\Big)^{\gamma d}\bigg)T\\
        &\stackrel{(a)}{=}\bigg(1-\Big(1-\frac{\gamma d}{T}+O\Big(\Big(\frac{\gamma d}{T}\Big)^2\Big)\bigg)T\\
        &=\bigg(\frac{\gamma d}{T}-O\Big(\Big(\frac{\gamma d}{T}\Big)^2\Big)\bigg)T\\
        &\stackrel{(b)}{=}\gamma d(1-\delta_n),
    \end{align}
    where (a) applies a second order Taylor expansion, and in (b) we introduce $\delta_n\in O\big(\frac{\gamma d}{T}\big)$.  Let $Z_1,Z_2,\dots,Z_{\gamma d}$ be the labels of the selected coupons and $W(\gamma d)=f(Z_1,\dots,Z_{\gamma d})$ be the number of distinct coupons. We have the bounded property difference property
    \begin{align}
        |f(Z_1,\dots,Z_j,\dots,Z_{\gamma d})-f(Z_1,\dots,\widehat{Z}_j,\dots,Z_{\gamma d})|\leq1
    \end{align}
    for any $j,Z_1,Z_2,\dots,Z_{\gamma d}$, and $Y'_j$, since the largest difference we can make is swapping a distinct coupon $Z_j$ for a non-distinct coupon $Z'_j$, or vice versa. McDiarmid's inequality \cite{McD89} gives
    \begin{align}
        \mathbb{P}(|f(Z_1,Z_2,\dots,Z_{\gamma d})-\E[f(Z_1,Z_2,\dots,Z_{\gamma d})]|\geq\delta)\leq2\exp\Big(-\frac{2\delta^2}{\gamma d}\Big).
    \end{align}
    Setting $\delta=(\gamma d)^{2/3}$, we get the desired result.
\end{proof}
Let $W^{(\mathcal{D}\setminus i)}$ and $W^{(\mathcal{D}\setminus i,j)}$ be the total number of positive tests containing at least one item in $\mathcal{D}\setminus\{i\}$, and the total number of positive tests containing at least one item in $\mathcal{D}\setminus\{i,j\}$ respectively. We then immediately obtain the following two corollaries.

\begin{corollary} \label{cor:concentration_pos_tests_d-1}
    When making $\gamma(d-1)\in o(T)$ draws with replacement from a total of $T$ coupons, the total number of distinct coupons $W^{(\mathcal{D}\setminus i)}$ satisfies
    \begin{align}
        \mathbb{P}\big( |W^{(\mathcal{D}\setminus i)}-\gamma(d-1)(1-\delta_n^{(1)})|\geq(\gamma(d-1))^{2/3} \big) \leq 2\exp(-2(\gamma(d-1))^{1/3}),
    \end{align}
    where $\delta_n^{(1)}\in O\big(\frac{\gamma d}{T}\big)$.
\end{corollary}

\begin{corollary} \label{cor:concentration_pos_tests_d-2}
    When making $\gamma(d-2)\in o(T)$ draws with replacement from a total of $T$ coupons, the total number of distinct coupons $W^{(\mathcal{D}\setminus i,j)}$ satisfies
    \begin{align}
        \mathbb{P}\big( |W^{(\mathcal{D}\setminus i,j)}-\gamma(d-2)(1-\delta_n^{(2)})|\geq(\gamma(d-2))^{2/3} \big) \leq 2\exp(-2(\gamma(d-2))^{1/3}),
    \end{align}
    where $\delta_n^{(2)}\in O\big(\frac{\gamma d}{T}\big)$.
\end{corollary}

\subsection{Proof of Theorem \ref{thm:gamma_lowerbound_ours_nearconstantwtdesign} (Converse for $\gamma$-Divisible Items)} \label{sec:gamma_lowerbound_ours_proof_SSS}

Throughout the proof, we condition on a fixed but otherwise arbitrary defective set $\mathcal{D}$.
We consider the event that some defective item $i\in\mathcal{D}$ is masked by the other defective items $\mathcal{D}\setminus\{i\}$, which leads to the event $|\widehat{\mathcal{D}}_{\text{SSS}}| < d$ \cite{Ald14a}. Hence, writing $A_i$ for the event that item $i\in\mathcal{D}$ is masked by $\mathcal{D}\setminus\{i\}$, de Caen's lower bound on a union \cite{Dec97} gives
\begin{align}
    \mathbb{P}( |\widehat{\mathcal{D}}_{\text{SSS}}| < d )
    \geq\mathbb{P}\bigg(\bigcup_{i\in\mathcal{D}}A_i\bigg)
    \geq\sum_{i\in\mathcal{D}}\frac{\mathbb{P}(A_i)^2}{\mathbb{P}(A_i)+\sum_{j\in\mathcal{D}\setminus\{i\}}\mathbb{P}(A_i\cap A_j)}. \label{eq:SSS_err_prob_bound}
\end{align}
We proceed by bounding the numerator and denominator separately.

\paragraph{Bounding the Numerator of \eqref{eq:SSS_err_prob_bound}}

Fixing the index $i$ of some defective item, we note that conditioned on $W^{(\mathcal{D}\setminus i)}=w$, the event $A_i$ occurs if each test that item $i$ occurs in is contained in the $w$ ``already hit'' tests. Hence, for any constant $c_1>0$, we have
\begin{align}
    \mathbb{P}(A_i)&=\sum_{w}\mathbb{P}(A_i|W^{(\mathcal{D}\setminus i)}=w)\mathbb{P}(W^{(\mathcal{D}\setminus i)}=w)\\
    &=\sum_{w}\Big(\frac{w}{T}\Big)^{\gamma}\mathbb{P}(W^{(\mathcal{D}\setminus i)}=w)\\
    &\geq\sum_{w\geq c_1\gamma(d-1)}\Big(\frac{w}{T}\Big)^{\gamma}\mathbb{P}(W^{(\mathcal{D}\setminus i)}=w)\\
    &\geq\sum_{w\geq c_1\gamma(d-1)}\Big(\frac{c_1\gamma(d-1)}{T}\Big)^{\gamma}\mathbb{P}(W^{(\mathcal{D}\setminus i)}=w)\\
    &=\Big(\frac{c_1\gamma(d-1)}{T}\Big)^{\gamma}\mathbb{P}(W^{(\mathcal{D}\setminus i)}\geq c_1\gamma(d-1)). \label{eq:SSS_err_prob_numerator}
\end{align}

\paragraph{Bounding the Denominator of \eqref{eq:SSS_err_prob_bound}}

We first derive a bound on $\mathbb{P}(A_i\cap A_j|W^{(\mathcal{D}\setminus i,j)}=w)$ that holds for any given $w = \Theta(\gamma d)$ (an event that holds with high probability by Corollary \ref{cor:concentration_pos_tests_d-2}), by suitably adapting the arguments of the unconstrained setting \cite{Joh16}.

For this part (and only this part), we represent columns of $\mathsf{X}$ corresponding to items $i$ and $j$ by lists, $\mathcal{T}_i=\{t_{i1},\dots,t_{i\gamma}\}$ and $\mathcal{T}_j=\{t_{j1},\dots,t_{j\gamma}\}$. Each list entry is obtained by choosing $t\in\{1,\dots,T\}$ uniformly at random with replacement, so duplicates may occur. Without loss of generality, we assume that the $w$ tests containing items from $\mathcal{D}\setminus\{i,j\}$ are those indexed by $1,\dots,w$. Any given list occurs with probability $1/T^{\gamma}$. Letting $\mathscr{A}_i$ be the set of list pairs $(\mathcal{T}_i,\mathcal{T}_j)$ under which the event $A_i$ occurs, and similarly for $\mathscr{A}_j$, we have
\begin{align}
    \mathbb{P}(A_i\cap A_j|W^{(\mathcal{D}\setminus i,j)}=w)=\frac{N_{ij}}{T^{2\gamma}}, \label{eq:conditional_prob_joint_A}
\end{align}
where 
\begin{align}
    N_{ij}=\sum_{\mathcal{T}_i}\sum_{\mathcal{T}_j}\mathds{1}\{(\mathcal{T}_i,\mathcal{T}_j)\in\mathscr{A}_i\cap\mathscr{A}_j\}
\end{align}
is the number of pairs of lists in $\mathscr{A}_i\cap\mathscr{A}_j$. Here the sets $\mathscr{A}_i$ and $\mathscr{A}_j$ implicitly depend on $w$. To bound $N_{ij}$, we separately consider the number of ``new positive tests'' caused by items $i$ and $j$; that is, not among the first $w$. Specifically, letting $N_{ij}(l)$ be defined as above with the summation limited to the case that there are $l$ such new positive tests, we have
\begin{align}
N_{ij}=\sum_{l=0}^{\gamma}N_{ij}(l),
\end{align}
where the summation goes up to $\gamma$ due to the fact that any new positive test containing $i$ must also contain $j$ and vice versa; otherwise, the masking under consideration would not occur.

To bound $N_{ij}(l)$, we consider the following procedure for choosing the lists:
\begin{itemize}
\item From $T-w$ tests, choose $l$ of them to be the new defective tests. This is one of ${T-w\choose l}$ options.
\item For both $i$ and $j$, assign one list index from $\{1,\dots,\gamma\}$ to each of the $l$ new defective tests. This is at most $\gamma^l$ options each, for $\gamma^{2l}$ in total.
\item For both $i$ and $j$, the remaining $\gamma-l$ list entries are chosen arbitrarily from the $w+l$ positive tests. This is $(w+l)^{\gamma-l}$ options each, for $(w+l)^{2(\gamma-l)}$ in total.
\end{itemize}
Combining these terms gives
\begin{align}
N_{ij}(l)&\leq{T-w\choose l}\cdot\gamma^{2l}\cdot(w+l)^{2(\gamma-l)}\\
&\leq(T-w)^l\cdot\gamma^{2l}\cdot(w+\gamma)^{2(\gamma-l)}\\
&=(w+\gamma)^{2\gamma}\cdot\Big(\frac{\gamma^2(T-w)}{(w+\gamma)^2}\Big)^l. \label{eq:bracketed_term}
\end{align}
Under the assumption that $w\geq c_1\gamma(d-2)$, the bracketed term 
$\frac{\gamma^2(T-w)}{(w+\gamma)^2}$ is less than any fixed $\epsilon_1>0$ for sufficiently large $n$. To see this, recall that $T\in\Theta(\gamma d\cdot d^{1/\gamma})=\Theta(\gamma n^{\theta+\theta/\gamma})$ and $w\in\Theta(\gamma d)=\Theta(\gamma n^\theta)$. By substituting the scaling regime for $T$ and $w$ into the bracketed term above and taking the log, we get
\begin{align}
\log\frac{\gamma^2(T-w)}{(w+\gamma)^2}
&=\log\frac{\gamma^2(\Theta(1)\gamma n^{\theta+\theta/\gamma}-\Theta(1)\gamma n^{\theta})}{(\Theta(1)\gamma n^{\theta}+\gamma)^2}\\
&\leq \log\frac{\Theta(1)\gamma^3n^{\theta+\theta/\gamma}}{\gamma^2n^{2\theta}}\\
&=\log \Theta(1) +\log(\gamma n^{\theta/\gamma-\theta})\\
&=\log\Theta(1) +\log\gamma+\Big(\frac{\theta}{\gamma}-\theta\Big)\log n.
\end{align}
%where (a) is by using $\Theta(1)\gamma n^{\theta+\theta/\gamma}-\Omega(1)\gamma n^{\theta}\leq\Theta(1)\gamma n^{\theta+\theta/\gamma}$ and $\Omega(1)\gamma n^{\theta}+\gamma\geq\Omega(1)\gamma n^{\theta}$. 
This expression tends to $-\infty$ for any $\gamma>1$, since $\gamma\in o(\log n)$. This implies that the bracketed term in \eqref{eq:bracketed_term} satisfies $\frac{\gamma^2(T-w)}{(w+\gamma)^2}\rightarrow0$ as $n\rightarrow\infty$, and is therefore less than any given $\epsilon_1 > 0$ for sufficiently large $n$. 

Summing over $l$, we obtain
\begin{align}
N_{ij}&\leq\sum_{l=0}^{\gamma}(w+\gamma)^{2\gamma}\cdot\Big(\frac{\gamma^2(T-w)}{(w+\gamma)^2}\Big)^l\\
&\leq(w+\gamma)^{2\gamma}\cdot\sum_{l=0}^\infty\epsilon_1^l\\
&=(w+\gamma)^{2\gamma}\cdot\frac{1}{1-\epsilon_1},
\end{align}
and substituting into \eqref{eq:conditional_prob_joint_A}, we obtain
\begin{align}
\mathbb{P}(A_i\cap A_j|W^{(\mathcal{D}\setminus i,j)}=w)
\leq\Big(\frac{w+\gamma}{T}\Big)^{2\gamma}\cdot\frac{1}{1-\epsilon_1}.
\end{align}
Now, for any $c_1,c_2>0$, we have
\begin{align}
\sum_{j\in\mathcal{D}\setminus\{i\}}\mathbb{P}(A_i\cap A_j)
&=(d-1)\sum_{w}\mathbb{P}(A_i\cap A_j|W^{(\mathcal{D}\setminus i,j)}=w)\mathbb{P}(W^{(\mathcal{D}\setminus i,j)}=w)\\
&\leq\frac{d-1}{1-\epsilon_1}\sum_{c_1\gamma(d-2)\leq w\leq c_2\gamma(d-2)}\Big(\frac{w+\gamma}{T}\Big)^{2\gamma}\mathbb{P}(W^{(\mathcal{D}\setminus i,j)}=w) \nonumber\\
&\quad\quad+(d-1)\mathbb{P}(W^{(\mathcal{D}\setminus i,j)}\notin[c_1\gamma(d-2),c_2\gamma(d-2)])\\
&\leq\frac{d-1}{1-\epsilon_1}\Big(\frac{c_2\gamma(d-2)+\gamma}{T}\Big)^{2\gamma}\mathbb{P}(c_1\gamma(d-2)\leq W^{(\mathcal{D}\setminus i,j)}\leq c_2\gamma(d-2)) \nonumber\\
&\quad\quad+(d-1)\mathbb{P}(W^{(\mathcal{D}\setminus i,j)}<c_1\gamma(d-2))+(d-1)\mathbb{P}(W^{(\mathcal{D}\setminus i,j)}>c_2\gamma(d-2)). \label{eq:SSS_err_prob_denominator}
\end{align}

\paragraph{Combining the Two Terms}

In accordance with Corollaries \ref{cor:concentration_pos_tests_d-1} and \ref{cor:concentration_pos_tests_d-2}, we choose $c_1$ and $c_2$ in \eqref{eq:SSS_err_prob_numerator} and \eqref{eq:SSS_err_prob_denominator} as follows:
\begin{align}
    c_1&=\min\bigg\{1-\delta^{(1)}_n-\frac{1}{(\gamma(d-1))^{1/3}},1-\delta^{(2)}_n-\frac{1}{(\gamma(d-2))^{1/3}}\bigg\}\\
    &\ge 1-\delta^{(3)}_n \\
    c_2&=1-\delta^{(2)}_n+\frac{1}{(\gamma(d-2))^{1/3}},
\end{align}
where $\delta_n^{(3)} \in O\big(\frac{\gamma d}{T}\big) + O\big( \frac{1}{(\gamma d)^{1/3}} \big) \rightarrow0$ (since $\delta^{(1)}_n$ and $\delta^{(2)}_n$ are both $O\big(\frac{\gamma d}{T}\big) $). 
%With our choices, the concentration results for both $W^{(\mathcal{D}\setminus i)}$ and $W^{(\mathcal{D}\setminus i,j)}$ in Corollary \ref{cor:concentration_pos_tests_d-1} and Corollary \ref{cor:concentration_pos_tests_d-2} respectively will hold. 
We also introduce
\begin{align}
    c_3=1+\Big(\frac{d-2}{d-1}\Big)(-\delta^{(2)}_n+(\gamma(d-2))^{-1/3}), \label{eq:c_3}
\end{align}
and note the useful fact
\begin{align}
    c_3\gamma(d-1)=c_2\gamma(d-2)+\gamma, \label{eq:c_3_fact}
\end{align}
which will be used later.

The concentration results from Corollaries \ref{cor:concentration_pos_tests_d-1} and \ref{cor:concentration_pos_tests_d-2} imply that
\begin{align}
    \mathbb{P}(A_i)&\stackrel{(a)}{\geq} \Big(\frac{c_1\gamma(d-1)}{T}\Big)^{\gamma}(1-o(1)) \label{eq:de_caen_sub_bound_1}\\
    \sum_{j\in\mathcal{D}\setminus\{i\}}\mathbb{P}(A_i\cap A_j)&\stackrel{(b)}{\leq}\frac{d-1}{1-\epsilon_1}\Big(\frac{c_2\gamma(d-2)+\gamma}{T}\Big)^{2\gamma}(1-o(1)) +o(1), \label{eq:de_caen_sub_bound_3}
\end{align}
where:
\begin{itemize}
    \item (a) follows by substituting our chosen $c_1$ into the $\mathbb{P}(\cdot)$ part of \eqref{eq:SSS_err_prob_numerator} and applying the concentration result in Corollary \ref{cor:concentration_pos_tests_d-1};
    \item (b) follows by substituting our chosen $c_1$ and $c_2$ into the $\mathbb{P}(\cdot)$ parts of \eqref{eq:SSS_err_prob_denominator}, and then apply the concentration results in Corollary \ref{cor:concentration_pos_tests_d-2}.
\end{itemize}
In these steps, we also used the fact that $2(d-1)\exp(-(\gamma(d-2))^{1/3})\rightarrow0$ since $\gamma\geq1$ and $d\rightarrow\infty$. 

In addition to \eqref{eq:de_caen_sub_bound_1}, we have the simple upper bound
\begin{align}
    \mathbb{P}(A_i)&\leq\Big(\frac{\gamma(d-1)}{T}\Big)^{\gamma}, \label{eq:de_caen_sub_bound_2}
\end{align}
which holds since the number of positive tests that contain at least one item in $\mathcal{D}\setminus\{i\}$ is trivially at most $\gamma(d-1)$. Substituting \eqref{eq:de_caen_sub_bound_1}--\eqref{eq:de_caen_sub_bound_2} into \eqref{eq:SSS_err_prob_bound}, we obtain
\begin{align}
    \mathbb{P}( |\widehat{\mathcal{D}}_{\text{SSS}}| < d )
    &\geq\sum_{i\in\mathcal{D}}\frac{\Big(\frac{c_1\gamma(d-1)}{T}\Big)^{2\gamma}(1-o(1))}{\Big(\frac{\gamma(d-1)}{T}\Big)^{\gamma}+\frac{d-1}{1-\epsilon_1}\Big(\frac{c_2\gamma(d-2)+\gamma}{T}\Big)^{2\gamma}(1-o(1))+o(1)}\\
    &\stackrel{(a)}{=}\frac{d\Big(\frac{\gamma(d-1)}{T}\Big)^{2\gamma}c_1^{2\gamma}(1-o(1))}{\Big(\frac{\gamma(d-1)}{T}\Big)^{\gamma}+\frac{d-1}{1-\epsilon_1}\Big(\frac{\gamma(d-1)}{T}\Big)^{2\gamma}c^{2\gamma}_3(1-o(1))+o(1)}\\
    &=\frac{d\Big(\frac{\gamma(d-1)}{T}\Big)^{\gamma}c_1^{2\gamma}(1-o(1))}{1+\frac{d-1}{1-\epsilon_1}\Big(\frac{\gamma(d-1)}{T}\Big)^{\gamma}c_3^{2\gamma}(1-o(1))+o(1)}\\
    &=\frac{d\Big(\frac{c^2_1\gamma(d-1)}{T}\Big)^{\gamma}(1-o(1))}{1+\frac{d-1}{1-\epsilon_1}\Big(\frac{c^2_3\gamma(d-1)}{T}\Big)^{\gamma}(1-o(1))+o(1)}, \label{eq:de_caen_err_bound_derived}
\end{align}
where (a) is by applying \eqref{eq:c_3_fact} in the denominator. 

In the following, it will be convenient to work with the following choice of $T$:
\begin{equation}
    T=\gamma d^{1/\gamma}(d-1)(c^2_3)(1-\epsilon_2). \label{eq:T_choice_new}
\end{equation}
Since we consider $d \to \infty$ and $c_2 = 1 - o(1)$, this choice is consistent with \eqref{eq:T_choice} for some $\epsilon_2 = \zeta + o(1)$.
Substituting \eqref{eq:T_choice_new} into \eqref{eq:de_caen_err_bound_derived}, we get
\begin{align}
    \mathbb{P}( |\widehat{\mathcal{D}}_{\text{SSS}}| < d )
    &\geq\frac{d\Big(\frac{c^2_1\gamma(d-1)}{\gamma d^{1/\gamma}(d-1)(c^2_3)(1-\epsilon_2)}\Big)^{\gamma}(1-o(1))}{1+\frac{d-1}{1-\epsilon_1}\Big(\frac{c^2_3\gamma(d-1)}{\gamma d^{1/\gamma}(d-1)(c^2_3)(1-\epsilon_2)}\Big)^{\gamma}(1-o(1))+o(1)}\\
    &=\frac{\big(\frac{c_1}{c_3}\big)^{2\gamma}\big(\frac{1}{1-\epsilon_2}\big)^\gamma(1-o(1))}{1+\frac{d-1}{d(1-\epsilon_1)}\big(\frac{1}{1-\epsilon_2}\big)^\gamma(1-o(1))+o(1)}\\
    &\stackrel{(a)}{=}\frac{\big(1-\delta^{(2)}_n-(\gamma(d-2))^{-1/3}\big)^{2\gamma}}{\Big(1+\big(\frac{d-2}{d-1}\big)(-\delta^{(2)}_n+(\gamma(d-2))^{-1/3})\Big)^{2\gamma}}\cdot\frac{1-o(1)}{(1-\epsilon_2)^\gamma+\frac{d-1}{d(1-\epsilon_1)}}\\
    &\stackrel{(b)}{=}\frac{1-o(1)}{(1-\epsilon_2)^\gamma+\frac{d-1}{d(1-\epsilon_1)}},
\end{align}
where (a) follows by substituting $c_1$ and $c_3$, and in (b),we note that both the numerator and denominator are in $\big(1-O\big(\frac{1}{d^{1/\gamma}}\big)\pm O\big(\frac{1}{(\gamma d)^{1/3}}\big)\big)^{2\gamma}=\big(1-O\big(\frac{1}{d^{1/\gamma}}\big)\big)^{2\gamma}$, and then apply Lemma \ref{lem:(1-o(1))^gamma->1}.  The proof is concluded by recalling that $\epsilon_1$ may be arbitrarily small, and $\epsilon_2 = \zeta + o(1)$.

\paragraph{Bounding the COMP Error Probability}

We use the fact that COMP fails if and only if at least one non-defective item is masked by $\mathcal{D}$, and denote the associated error probability by $\mathbb{P}^{\text{COMP}}(\text{err})$. 

Recall from Lemma \ref{lem:concentration_pos_tests} that the number $W^{(\mathcal{D})}$ of positive tests lies in $[\gamma d(1-\delta_n)-(\gamma d)^{2/3},\gamma d(1-\delta_n)+(\gamma d)^{2/3}]$ with high probability, where $\delta_n \in O\big( \frac{\gamma d}{T} \big) \to 0$.  For any $w$ in this range and any non-defective item $i$, we have
\begin{align}
    \mathbb{P}(\text{$i$ masked by $\mathcal{D}$}|W^{(\mathcal{D})}=w)
    &\geq \Big(\frac{\gamma d(1-\delta_n^{(4)})}{T}\Big)^{\gamma}, \label{eq:prob_err_COMP1}
\end{align}
where $\delta^{(4)}_n \in O\big(\frac{\gamma d}{T}\big)=O\big(\frac{1}{d^{1/\gamma}}\big) \to 0$, since $i$ is masked if and only if all of its $\gamma$ tests are those among the $w$ positive ones.

Using \eqref{eq:prob_err_COMP1}, we derive an upper bound on $\mathbb{P}^{\text{COMP}}(\text{suc}) = 1 - \mathbb{P}^{\text{COMP}}(\text{err})$:
\begin{align}
    \mathbb{P}^{\text{COMP}}(\text{suc})
    &=\sum_w\mathbb{P}(W^{(\mathcal{D})}=w)\mathbb{P}^{\text{COMP}}(\text{suc}|W^{(\mathcal{D})}=w)\\
    &=\sum_w\mathbb{P}(W^{(\mathcal{D})}=w)\mathbb{P}(\text{all $i$ not masked by $\mathcal{D}$}|W^{(\mathcal{D})}=w)\\
    &=\sum_w\mathbb{P}(W^{(\mathcal{D})}=w)\Big(1-\mathbb{P}(\text{$i$ masked by $\mathcal{D}$}|W^{(\mathcal{D})}=w)\Big)^{n-d}\\
    % &\stackrel{(a)}{\leq}\sum_{\substack{w\in[\gamma d(1-\delta_n)-(\gamma d)^{2/3} \\ ,\gamma d(1-\delta_n)+(\gamma d)^{2/3}]}}\mathbb{P}(W^{(\mathcal{D})}=w)\Big(1-\mathbb{P}(\text{$i$ masked by $\mathcal{D}$}|W^{(\mathcal{D})}=w)\Big)^{n-d}  +2\exp(-2(\gamma d)^{\frac{1}{3}}) \\
    % &\quad+\sum_{\substack{w\notin[\gamma d(1-\delta_n)-(\gamma d)^{2/3} \\ ,\gamma d(1-\delta_n)+(\gamma d)^{2/3}]}}\mathbb{P}(W^{(\mathcal{D})}=w)\Big(1-\mathbb{P}(\text{$i$ masked by $\mathcal{D}$}|W^{(\mathcal{D})}=w)\Big)^{n-d}\\
    &\stackrel{(a)}{\leq}\bigg(1-\Big(\frac{\gamma d(1-\delta^{(4)}_n)}{T}\Big)^{\gamma}\bigg)^{n-d}
    +2\exp(-2(\gamma d)^{\frac{1}{3}})\\
    &\stackrel{(b)}{=}\bigg(1-\Big(\frac{\gamma d(1-\delta^{(4)}_n)}{T}\Big)^{\gamma}\bigg)^{n-d}+o(1), \label{eq:COMP_final}
\end{align}
%(a) follows by using Lemma \ref{lem:concentration_pos_tests} to upper bound $\mathbb{P}(W^{(\mathcal{D})}\notin[\gamma d(1-\delta_n)-(\gamma d)^{2/3},\gamma d(1-\delta_n)+(\gamma d)^{2/3}])$, 
where (a) follows from Lemma \ref{lem:concentration_pos_tests} and \eqref{eq:prob_err_COMP1}, and (b) uses $2\exp(-2(\gamma d)^{\frac{1}{3}})\rightarrow0$ since $\gamma\geq1$ and $d\rightarrow\infty$.  It follows from \eqref{eq:COMP_final} that
    % and upper bounding $\mathbb{P}(W^{(\mathcal{D})}\in[\gamma d(1-\delta_n)-(\gamma d)^{2/3},\gamma d(1-\delta_n)+(\gamma d)^{2/3}])$ by 1. For the second term, we first upper bound $\big(1-\mathbb{P}(\text{$i$ masked by $\mathcal{D}$}|W^{(\mathcal{D})}=w)\big)^{n-d}$ by 1. Next, we apply (b) is obtained by using the fact that $2\exp(-2(\gamma d)^{\frac{1}{3}})\rightarrow0$ since $\gamma\geq1$ and $d\rightarrow\infty$. Applying the above results, we provide a lower bound on $\mathbb{P}^{\text{COMP}}(\text{err})$:
\begin{align}
    \mathbb{P}^{\text{COMP}}(\text{err})
    &=1-\mathbb{P}^{\text{COMP}}(\text{suc})\\
    &\geq1-\bigg(1-\Big(\frac{\gamma d(1-\delta^{(4)}_n)}{T}\Big)^{\gamma}\bigg)^{n-d}-o(1)\\
    &=1-\bigg(1-\Big(\frac{\gamma d}{T}\Big)^{\gamma}(1-\delta^{(4)}_n)^{\gamma}\bigg)^{n-d}-o(1). \label{eq:prob_err_COMP2}
\end{align}
Recall that we choose $T=\gamma d^{\frac{1}{\gamma}}(d-1)(c_3^2)(1-\epsilon_2)$ for some constant $\epsilon_2>0$ (see \eqref{eq:T_choice_new}); substituting into \eqref{eq:prob_err_COMP2}, we obtain
\begin{align}
    \mathbb{P}^{\text{COMP}}(\text{err})
    &\geq1-\bigg(1-\Big(\frac{\gamma d}{\gamma d^{1/\gamma}(d-1)(c_3^2)(1-\epsilon_2)}\Big)^{\gamma}(1-\delta_n^{(4)})^\gamma\bigg)^{n-d}-o(1)\\
    &=1-\bigg(1-\frac{1}{d}\Big(1+\frac{1}{d-1}\Big)^{\gamma}\frac{(1-\delta_n^{(4)})^\gamma}{c_3^{2\gamma}(1-\epsilon)^\gamma}\bigg)^{n-d}-o(1)\\
    &\stackrel{(a)}{=}1-\bigg(1-\frac{(1+o(1))}{d(1-\epsilon_2)^{\gamma}}\bigg)^{n-d}-o(1)\\
    &\stackrel{(b)}{=}1-\exp\Big(-\frac{(n-d)(1+o(1))}{d(1-\epsilon_2)^{\gamma}}\Big)-o(1)\\
    &=1-\exp\Big(-\frac{n(1+o(1))}{d(1-\epsilon_2)^{\gamma}}\Big)-o(1), \label{eq:COMP_final2}
\end{align}
where:
\begin{itemize}
    \item (a) follows by applying $\big(1+\frac{1}{d-1}\big)^{\gamma}=1+o(1)$ (since $d\gg\gamma$ under our considered scaling laws), as well as noting that both $c_3$ and $1-\delta_n^{(4)}$ behave as $1-O\big(\frac{1}{d^{1/\gamma}}\big)$, and applying Lemma \ref{lem:(1-o(1))^gamma->1} to get $c_3^{2\gamma} = 1-o(1)$ and $(1-\delta_n^{(4)})^\gamma\in1-o(1)$. 
    \item (b) follows by first noting that $\frac{1+o(1)}{d(1-\epsilon_2)^{\gamma}}\in o(1)$ (proved shortly), and then applying $1-z=e^{-z(1+o(1))}$ when $z\in o(1)$. To see why $\frac{1+o(1)}{d(1-\epsilon_2)^{\gamma}}\in o(1)$, we can take the log of the denominator and substitute the respective scaling regimes to get $\Theta(\theta \log n)+(\log n)^c\log(1-\epsilon_2) = \Theta(\log n)$.
\end{itemize}
The right-hand side of \eqref{eq:COMP_final2} approaches one as $n\rightarrow\infty$, since $\exp\big(-\frac{n}{d(1-\epsilon_2)^{\gamma}}\big)\rightarrow0$ by the assumption that $d\in\Theta(n^{\theta})$ with $\theta\in(0,1)$.

\subsection{Proof of Theorem \ref{thm:gamma_upperbound_nonadap} (DD Performance)} \label{sec:gamma_upperbound_nonadap}

We again condition on a fixed but otherwise arbitrary defective set $\mathcal{D}$.
We observe that first and second steps recover $\mathcal{D}$ correctly when each defective item $i$ is not masked (see Definition \ref{def:mask}) by $ \mathcal{PD} \setminus\{i\}$. Hence, we want to derive a bound on $T$ ensuring that the probability of each defective item $i$ being masked by $ \mathcal{PD} \setminus\{i\}$ is vanishing. Each defective item $i$ is masked by $ \mathcal{PD} \setminus\{i\}$ only when the number of collisions (see Definition \ref{def:collision}) between $i$ and $\mathcal{PD}$ is $\gamma$. Since $\mathcal{PD}$ can be split into two sets $\mathcal{D}$ and $\mathcal{PD}\setminus\mathcal{D}$, we can consider the number of collisions between $i$ and each of these two sets separately. This motivates the main steps of our proof:
\begin{enumerate}
    \item We derive a concentration result on the number of non-defective items in  $\mathcal{PD}$.
    \item We derive a bound on $T$ such that there is a low probability of any defective item $i$ incurring ``too many'' collisions with $\mathcal{D}\setminus\{i\}$.
    \item Conditioned on not having too many collisions between defectives in the sense of the previous item, we derive a bound on $T$ such that there is also a low probability of any defective item $i$ having every one of its ``collision-free'' tests contain at at least one item from $ \mathcal{PD} \setminus\mathcal{D}$.
    \item Taking the maximum between the two bounds on $T$ gives us the required number of tests.
\end{enumerate}
We proceed by analysing the two steps of the DD algorithm separately.

\paragraph{Analysis of the First Step}

Let $G=| \mathcal{PD} \setminus\mathcal{D}|$ denote the number of non-defective items in  $\mathcal{PD}$, where $G=\sum_{i=1}^{n-d}G_i$ with $G_i\in\{0,1\}$ being the corresponding indicator variable for a single non-defective. Conditioned on the number of positive tests being $W^{(\mathcal{D})}=w^{(\mathcal{D})}$, we have
\begin{align}
    (G|W^{(\mathcal{D})}=w^{(\mathcal{D})})\sim\text{Binomial}\bigg(n-d,\Big(\frac{w^{(\mathcal{D})}}{T}\Big)^{\gamma}\bigg),
\end{align}
where $\big(\frac{w^{(\mathcal{D})}}{T}\big)^{\gamma}\leq\big(\frac{\gamma d}{T}\big)^{\gamma}$ because the number of positive tests is at most $\gamma d$. Since $G_i \in \{0,1\}$, Bernstein's inequality gives the following for any $t > 0$:
\begin{align}
    \mathbb{P}\Big(\sum_{i=1}^{n-d}G_i>\mathbb{E}[G]+t\Big)
    &\leq\exp\bigg(\frac{-\frac{1}{2}t^2}{\sum_{i=1}^{n-d}\Var[G_i]+\frac{1}{3}t}\bigg)\\
    &\stackrel{(a)}{\leq}\exp\bigg(\frac{-\frac{1}{2}t^2}{\sum_{i=1}^{n-d}\mathbb{E}[G_i^2]+\frac{1}{3}t}\bigg)\\
    &\stackrel{(b)}{\leq}\exp\bigg(\frac{-\frac{1}{2}t^2}{\sum_{i=1}^{n-d}\mathbb{E}[G_i]+\frac{1}{3}t}\bigg)\\
    &\stackrel{(c)}{\leq}\exp\bigg(\frac{-\frac{1}{2}t^2}{\mathbb{E}[G]+\frac{1}{3}t}\bigg),
    \label{eq:Bernstein_bound_on_G}
\end{align}
where (a) uses $\Var[G_i]=\mathbb{E}[G_i^2]-(\mathbb{E}[G_i])^2\leq\mathbb{E}[G_i^2]$, (b) follows since $G_i\in\{0,1\}$ and hence $\mathbb{E}[G_i^2]=\mathbb{E}[G_i]$, and (c) is due to the linearity of expectation.  We will return to \eqref{eq:Bernstein_bound_on_G} and select $t$ later in the analysis.

\paragraph{Analysis of the Second Step}

Firstly, we want to show that the event in which the number of collisions between a chosen defective item $i$ and $\mathcal{D}\setminus\{i\}$ is ``close to $\gamma$'' (to be formalized later) is a rare event. It is easy to see that rearranging the columns of the test matrix only amounts to re-labeling items.  Hence, for clarity, we think of the test matrix as being rearranged such that the first $d$ columns are for the defective items, as shown in Figure \ref{fig:d_collide}.
\begin{figure}[t] 
  \centering
  \includegraphics[scale=0.4]{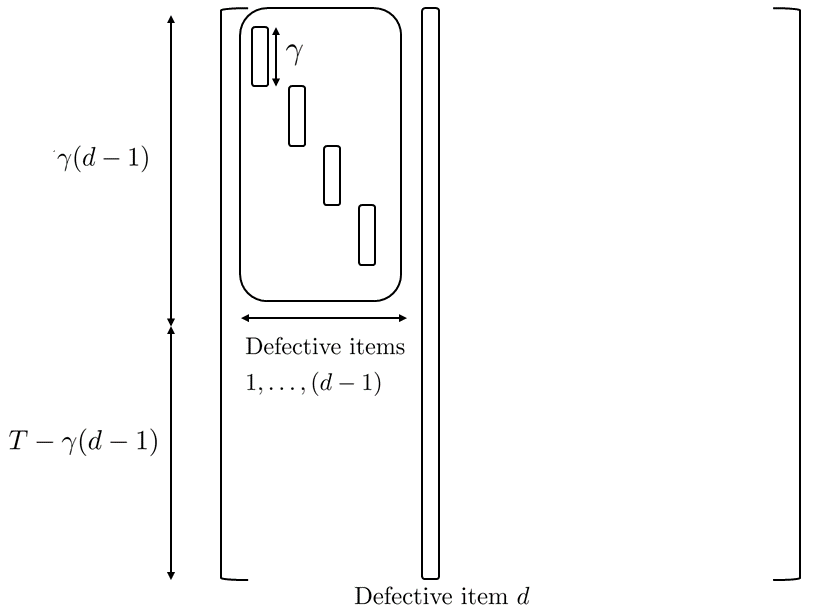}
  \caption{Rearranged test matrix to illustrate that collision within defectives are rare.  This illustration shows the extreme case that the first $d-1$ defective items are each in a distinct set of $\gamma$ tests.}
  \label{fig:d_collide}
\end{figure}
Referring to Figure \ref{fig:d_collide}, let $C_i$ be the number of collisions between a given defective item $i$ (with $i=d$ in Figure \ref{fig:d_collide}) and $\mathcal{D}\setminus\{i\}$. Recall that $W^{(\mathcal{D}\setminus i)}$ denotes the number of positive tests containing at least one item in $\mathcal{D}\setminus\{i\}$. Given $W^{(\mathcal{D}\setminus i)}=w^{(\mathcal{D}\setminus i)}$, we have
\begin{align}
    (C_i|W^{(\mathcal{D}\setminus i)}=w^{(\mathcal{D}\setminus i)})\sim\text{Binomial}\Big(\gamma,\frac{w^{(\mathcal{D}\setminus i)}}{T}\Big),
\end{align}
where $\frac{w^{(\mathcal{D}\setminus i)}}{T}\leq\frac{\gamma(d-1)}{T}<\frac{\gamma d}{T}$ because any $w^{(\mathcal{D}\setminus i)}$ is at most $\gamma(d-1)$. We want to show that $\mathbb{P}(C_i\geq\alpha_2\gamma)$ is small, where $\alpha_2\in(0,1)$. We first note that
\begin{align}
    \mathbb{P}(C_i=\alpha_2\gamma) &={\gamma\choose\alpha_2\gamma}\Big(\frac{w^{(\mathcal{D}\setminus i)}}{T}\Big)^{\alpha_2\gamma}(1-p)^{\gamma-\alpha_2\gamma}\\
    &\stackrel{(a)}{\leq}e^{\gamma H_2(\alpha_2)}\Big(\frac{\gamma d}{T}\Big)^{\alpha_2\gamma}.
\end{align}
where (a) is due to ${\gamma\choose\alpha_2\gamma}\leq e^{\gamma H_2(\alpha_2)}$ (where $H_2(\cdot)$ is the binary entropy function in nats), $\frac{w^{(\mathcal{D}\setminus i)}}{T}<\frac{\gamma d}{T}$, and $(1-p)^{\gamma-\alpha_2\gamma}\leq1$. We proceed to show that $\mathbb{P}(C_i\geq\alpha_2\gamma)$ behaves similarly to $\mathbb{P}(C_i=\alpha_2\gamma)$:
\begin{align}
    \mathbb{P}(C_i\geq\alpha_2\gamma)&=\sum_{k=\alpha_2\gamma}^{\gamma}\mathbb{P}(C_i=k)\\
    &\leq\sum_{k=\alpha_2\gamma}^{\gamma}e^{\gamma H_2(k/\gamma)}\Big(\frac{\gamma d}{T}\Big)^k\\
    &\stackrel{(a)}{\leq}e^{\gamma H_2(\max\{\alpha_2,\frac{1}{2}\})}\sum_{k=\alpha_2\gamma}^{\infty}\Big(\frac{\gamma d}{T}\Big)^k\\
    &\stackrel{(b)}{=}e^{\gamma H_2(\max\{\alpha_2,\frac{1}{2}\})}\frac{(\frac{\gamma d}{T})^{\alpha_2\gamma}}{1-(\frac{\gamma d}{T})}\\
    &\stackrel{(c)}{=}e^{\gamma H_2(\max\{\alpha_2,\frac{1}{2}\})}\Big(\frac{\gamma d}{T}\Big)^{\alpha_2\gamma}(1+o(1)),
\end{align}
where (a) uses the fact that $H_2(p)$ is increasing for $p\leq\frac{1}{2}$ and decreasing for $p\geq\frac{1}{2}$, (b) uses the sum to infinity of a geometric series, and (c) uses the fact that $\frac{\gamma d}{T}\in o(1)$. By the union bound, we have the following:
\begin{align}
    \mathbb{P}\bigg(\bigcup_{i=1}^d\{C_i\geq\alpha_2\gamma\}\bigg)\leq de^{\gamma H_2(\max\{\alpha_2,\frac{1}{2}\})}\Big(\frac{\gamma d}{T}\Big)^{\alpha_2\gamma}(1+o(1)). \label{eq:error_prob_2}
\end{align}
For \eqref{eq:error_prob_2} to approach zero, we consider the following condition for $T$, where $\beta_n$ is a slowly decaying term as $n\rightarrow\infty$:
\begin{align}
    e^{\gamma H_2(\max\{\alpha_2,\frac{1}{2}\})}\Big(\frac{\gamma d}{T}\Big)^{\alpha_2\gamma} &\leq\frac{\beta_n}{d} ,
    \label{eq:error_bound_term2}
\end{align}
which simplifies to
\begin{align}
    T&\geq\gamma de^{\frac{1}{\alpha_2}H_2(\max\{\alpha_2,\frac{1}{2}\})}\Big(\frac{d}{\beta_n}\Big)^{\frac{1}{\alpha_2\gamma}}. \label{eq:DD_tests_1}
\end{align}

Now, we study the probability of defective item $i$ not being in $\widehat{\mathcal{D}}$. We will first condition on the event that for any defective item $i$, the number of collisions between defective item $i$ and $\mathcal{D}\setminus\{i\}$ is not too high (to be formalized later). After conditioning, we consider the event where every test that includes defective item $i$, and no other defective item, contains at least one item from $ \mathcal{PD} \setminus\mathcal{D}$. This is equivalent to the event that defective item $i \notin \widehat{\mathcal{D}}$. We derive a bound on $T$ such that the probability of this event is vanishing. % We start by claiming that $T\in\Omega\big(\gamma d\big(\frac{n}{d}\big)^{1/\gamma}d^{1/\gamma^2}\big)$.

We condition on the following events:
\begin{enumerate}
    \item $\bigcap_{i=1}^d\{C_i<\alpha_2\gamma\}$. This occurs with high probability since we already ensured $\mathbb{P}\big( \bigcup_{i=1}^d\{C_i\geq\alpha_2\gamma\}\big) \rightarrow0$. % Thus, we condition on $\neg[\bigcup_{i=1}^d\{C_i\geq\alpha_2\gamma\}]\equiv\bigcap_{i=1}^d\{C_i<\alpha_2\gamma\}$ (by De Morgan's law).
    \item $W^{(\mathcal{D})}=\gamma d(1-\delta^{-}_n)$, where $\delta^{-}_n \in [\delta_n-(\gamma d)^{-1/3},\delta_n+(\gamma d)^{-1/3}]$ in accordance with Lemma \ref{lem:concentration_pos_tests}, and hence $\delta^{-}_n \in O\big(\frac{\gamma d}{T} + (\gamma d)^{-1/3} \big)$.  By the choice of $T$ in \eqref{eq:T_long}, this scaling on $\delta^{-}_n$ can be simplified to $\delta^{-}_n\in O\big(\frac{1}{(n/d)^{1/\gamma}d^{1/\gamma^2}}\big)=O\big(\frac{1}{n^{(1-\theta)/\gamma}}\big)$, as the $(\gamma d)^{-1/3}$ term is comparatively negligible.
\end{enumerate}
In addition, we condition on a fixed value of $G \le d$; this condition will be seen to hold with high probability once we choose $t$ in \eqref{eq:Bernstein_bound_on_G}.

 % For clarity, we repeat the event that we are interested in: given $\bigcap_{i=1}^d\{C_i<\alpha_2\gamma\}$ and $W^{(\mathcal{D})}=\gamma d(1-\delta^{-}_n)$, defective item $i$ is not in $\widehat{\mathcal{D}}$. The main steps to derive a bound on $T$ are as follows: we derive a bound on $G$ ensuring that the event is rare, which gives us our bound on $T$ by further applying the concentration results from the analysis of the first step.

We start by looking at a single defective item. Let $\tilde{\gamma}$ be the number of tests in which defective item $i$ is the only defective item. Since we conditioned on $\bigcap_{i=1}^d\{C_i<\alpha_2\gamma\}$, we have $\tilde{\gamma}\geq(1-\alpha_2)\gamma$. We want to find the probability that all $\tilde{\gamma}$ indices correspond to tests where at least one non-defective item in $\mathcal{PD}$ is also present. 
Without loss of generality, we assume that the $\tilde{\gamma}$ tests of interest are those labeled $1$ to $\tilde{\gamma}$. We let $A_i$ be the event that the positive test indexed by $i$ contains at least one non-defective item in $\mathcal{PD}$. 

To study the $A_i$ events, we first note that the non-defective test placements are independent of the defective ones, and recall that we condition on a fixed value of $G = |\mathcal{PD} \setminus \mathcal{D}|$ and a fixed number $\gamma d(1-\delta^-_n)$ of positive tests.  We consider the process of collecting $G\gamma$ ``coupons'' (placements into tests) corresponding to the non-defective items in $\mathcal{PD} \setminus \mathcal{D}$  Each coupon collected must correspond to a positive test, since otherwise the item would not be in $\mathcal{PD} \setminus \mathcal{D}$.  In addition, since the prior test placement distribution was uniform, the conditional distribution remains uniform, but is now only over the $\gamma d(1-\delta^-_n)$ positive tests.

Putting the above observations together, we consider a population of $\gamma d(1-\delta^-_n)$ coupons, the first $\tilde{\gamma}$ of which correspond to $i$ being the unique defective item.  We consider collecting $G\gamma$ coupons chosen uniformly with replacement. Then, the event $A_i$ is equivalent to the event that coupon $i$ is collected, yielding
\begin{align}
    \mathbb{P}(A_1,\dots,A_{\tilde{\gamma}})=\frac{\text{\#ways to collect coupons including all of the first $\tilde{\gamma}$ indices}}{\text{\#ways to collect all coupons}}, \label{eq:num_denom}
\end{align}
where the probability is implicitly conditioned on the events described above.

We will bound the numerator and denominator of \eqref{eq:num_denom} separately. For the denominator, we can think of listing the selected $G\gamma$ coupons in a vector, where each entry can be any of the $\gamma d(1-\delta^-_n)$ coupons from the population. This gives the following:
\begin{align}
    (\text{\#ways to collect all coupons})=(\gamma d)^{G\gamma}(1-\delta^-_n)^{G\gamma}.
\end{align}
For the numerator in \eqref{eq:num_denom}, we have
\begin{align}
    \parbox{15em}{(\centering\#ways to collect coupons\\including all of the first $\tilde{\gamma}$ indices)}% &\leq\bigg[\prod_{i=0}^{\tilde{\gamma}-1}(G\gamma-i)\bigg](\gamma d)^{G\gamma-\tilde{\gamma}}(1-\delta^-_n)^{G\gamma-\tilde{\gamma}}\\
    &\stackrel{(a)}{\leq} (G\gamma)^{\tilde{\gamma}}(\gamma d)^{G\gamma-\tilde{\gamma}}(1-\delta^-_n)^{G\gamma-\tilde{\gamma}} \label{eq:G_count}\\
    &=\Big(\frac{G}{d}\Big)^{\tilde{\gamma}}(\gamma d)^{G\gamma}(1-\delta^-_n)^{G\gamma-\tilde{\gamma}},
\end{align}
where (a) follows since each index in the set $\{1,2,\cdots,\tilde{\gamma}\}$ must occupy at least one position in the sequence of $G\gamma$ coupons; after assigning one such position to each index (in one of at most $(G\gamma)^{\tilde{\gamma}}$ ways), each of the remaining $G\gamma-\tilde{\gamma}$ positions can take any of the $\gamma d(1-\delta^-_n)$ indices. 

Combining the bounds on numerator and denominator, we have
\begin{align}
    \mathbb{P}(A_1,\dots,A_{\tilde{\gamma}}) &=\frac{\text{\#ways to collect coupons including all of the first $\tilde{\gamma}$ indices}}{\text{\#ways to collect all coupons}}\\
    &=\Big(\frac{G}{d}\Big)^{\tilde{\gamma}}(1-\delta^-_n)^{-\tilde{\gamma}}\\
    &\stackrel{(a)}{\leq}\Big(\frac{G}{d}\Big)^{(1-\alpha_2)\gamma}(1-\delta^-_n)^{-\tilde{\gamma}}\\
    &\stackrel{(b)}{=}\Big(\frac{G}{d}\Big)^{(1-\alpha_2)\gamma}(1+o(1)),
\end{align}
where (a) holds since $G\leq d$ and $(1-\alpha_2)\gamma\leq\tilde{\gamma}$, and (b) follows by recalling that $\delta^-_n\in O\big(\frac{1}{n^{(1-\theta)/\gamma}}\big)$ %, where the scaling regime is similar (differing only in the constant multiplicative factor of the power) to $O\big(\frac{1}{d^{1/\gamma}}\big)=O\big(\frac{1}{n^{\theta/\gamma}}\big)$, 
and applying Lemma \ref{lem:(1-o(1))^gamma->1}. 

According to the DD algorithm, the event where there exists a defective item not in $\widehat{\mathcal{D}}$ is equivalent to there existing a defective item where all its $\tilde{\gamma}$ indices are collected by the $G\gamma$ coupons. Applying union bound, the bound on the probability is as follows.
\begin{align}
    \mathbb{P}(\text{$\exists$ defective item not in $\widehat{\mathcal{D}}$})
    &\leq\mathbb{P}\bigg(\bigcup_{i=1}^d\{\text{all $\tilde{\gamma}$ indices of $i$ are collected}\}\bigg) \\
    &\leq d\Big(\frac{G}{d}\Big)^{(1-\alpha_2)\gamma}(1+o(1)). \label{eq:error_prob_4}
\end{align}
The bound approaches zero if $\big(\frac{G}{d}\big)^{(1-\alpha_2)\gamma}\leq\frac{\beta_n}{d}$ where $\beta_n$ is a slowly decaying term as $n\rightarrow\infty$. Rearranging, we obtain the following sufficient condition on $G$ to ensure that the right-hand side of \eqref{eq:error_prob_4} vanishes:
\begin{align}
    G\leq d\Big(\frac{\beta_n}{d}\Big)^{\frac{1}{(1-\alpha_2)\gamma}}. \label{eq:error_bound_term3}
\end{align}

\paragraph{Combining the two Steps}

We now combine the two steps to obtain our final bound on $T$, as well as studying the overall error probability. In accordance with \eqref{eq:error_bound_term3}, we define
\begin{align}
    G_{\max}=d\Big(\frac{\beta_n}{d}\Big)^{\frac{1}{(1-\alpha_2)\gamma}}.
\end{align}
Recall that $\E[G]=(n-d)\big(\frac{w^{(\mathcal{D})}}{T}\big)^{\gamma} \leq (n-d) \big(\frac{\gamma d}{T}\big)^\gamma$, and hence, $\mathbb{E}[G]\leq G_{\max}/2$ is guaranteed when
\begin{align}
    (n-d)\Big(\frac{\gamma d}{T}\Big)^{\gamma}\leq\frac{d}{2}\Big(\frac{\beta_n}{d}\Big)^{\frac{1}{(1-\alpha_2)\gamma}}, \label{eq:to_rearrange}
\end{align}
which we will shortly re-arrange to deduce a condition on $T$.
Then, setting $t=G_{\max}/2$ in our inequality in \eqref{eq:Bernstein_bound_on_G} gives
\begin{align}
    \mathbb{P}\Big(G>\mathbb{E}[G]+\frac{G_{\max}}{2}\Big)&\leq\exp\bigg(\frac{-\frac{1}{2}\big(\frac{G_{\max}}{2}\big)^2}{\mathbb{E}[G]+\frac{1}{3}\big(\frac{G_{\max}}{2}\big)}\bigg).
\end{align}
Applying that fact that $\mathbb{E}[G]\leq G_{\max}/2$, we get
\begin{align}
    \mathbb{P}(G>G_{\max})&\leq\exp\Big(-\frac{3}{16}G_{\max}\Big)
    =\exp\Big(-\frac{3d}{16}\bigg(\frac{\beta_n}{d}\Big)^{\frac{1}{(1-\alpha_2)\gamma}}\bigg),
    \label{eq:error_prob_1}
\end{align}
which approaches zero as long is $\beta_n$ does not decay too rapidly. By combining all the error probabilities in \eqref{eq:error_prob_1}, \eqref{eq:error_prob_2}, \eqref{eq:error_prob_4}, and Lemma \ref{lem:concentration_pos_tests} (with $\delta_n$ defined therein), we have
\begin{align}
    P_e
    &\leq\mathbb{P}(G>G_{\text{max}})+\mathbb{P}\bigg(\bigcup_{i=1}^d\{C_i\geq\alpha_2\gamma\}\bigg)
    +\mathbb{P}(\exists\text{ defective item not in $\widehat{\mathcal{D}}$}) \nonumber\\
    &\qquad+\mathbb{P}\Big(\delta^-_n\notin[\delta_n-(\gamma d)^{-1/3},\delta_n+(\gamma d)^{-1/3}]\Big) \\
    &\leq\exp\Big(-\frac{3d}{16}\bigg(\frac{\beta_n}{d}\Big)^{\frac{1}{(1-\alpha_2)\gamma}}\bigg)+de^{\gamma H_2(\max\{\alpha_2,\frac{1}{2}\})}\Big(\frac{\gamma d}{T}\Big)^{\alpha_2\gamma}(1+o(1)) \nonumber\\
    &\qquad+d\Big(\frac{G}{d}\Big)^{(1-\alpha_2)\gamma}(1+o(1))
    +2\exp(-2(\gamma d)^{1/3})\\
    &\stackrel{(a)}{\leq}\exp\Big(-\frac{3d}{16}\bigg(\frac{\beta_n}{d}\Big)^{\frac{1}{(1-\alpha_2)\gamma}}\bigg)
    +2\beta_n(1+o(1))
    +2\exp(-2(\gamma d)^{1/3}),
\end{align}
where (a) applies $d e^{\gamma H_2(\max\{\alpha_2,\frac{1}{2}\})}\big(\frac{\gamma d}{T}\big)^{\alpha_2\gamma}(1+o(1))\leq\beta_n(1+o(1))$ from \eqref{eq:error_bound_term2}, as well as $d\big(\frac{G}{d}\big)^{(1-\alpha_2)\gamma}(1+o(1))\leq\beta_n(1+o(1))$ from \eqref{eq:error_bound_term3}.

As for the number of tests, rearranging \eqref{eq:to_rearrange} gives
\begin{align}
    T& \geq 2^{1/\gamma}\gamma d\Big(\frac{n-d}{d}\Big)^{\frac{1}{\gamma}}\Big(\frac{d}{\beta_n}\Big)^{\frac{1}{(1-\alpha_2)\gamma^2}}.
    \label{eq:DD_tests_2}
\end{align}
Combining \eqref{eq:DD_tests_1} and \eqref{eq:DD_tests_2} gives
\begin{align}
    T&\geq\gamma d
    \max\bigg\{e^{\frac{1}{\alpha_2}H_2(\max\{\alpha_2,\frac{1}{2}\})}\Big(\frac{d}{\beta_n}\Big)^{\frac{1}{\alpha_2\gamma}}
    ,2^{1/\gamma}\Big(\frac{n-d}{d}\Big)^{\frac{1}{\gamma}}\Big(\frac{d}{\beta_n}\Big)^{\frac{1}{(1-\alpha_2)\gamma^2}}\bigg\},
\end{align}
which coincides with \eqref{eq:T_long} in the theorem statement.

\section{Conclusion} \label{ch:conclusion_and_future_work}

We have studied the problem of group testing with $\gamma$-divisible items (and, more briefly, $\rho$-sized tests).  In the adaptive setting, we characterized the optimal number of tests up to a multiplicative factor of $e^{1+o(1)}$ in broad scaling regimes, via both a strengthened counting-based converse and a novel adaptive splitting algorithm.  In the non-adaptive setting, we provided an algorithm-independent converse the near-constant tests-per-item design, and gave a strengthened achievability bound (essentially matching the converse in broad scaling regimes) via the DD algorithm.  An open challenge for future work would be to pursue optimal or near-optimal constant factors, as opposed to only optimality with respect to $\eta$ defined in \eqref{eq:eta}.

% use socreport.bst
%\bibliographystyle{socreport}
\bibliographystyle{IEEEtran}

% use references.bib

\bibliography{JS_References,references}

% \appendix

% \section{Proof}
% In this appendix, we present alternate, longer, but more interesting proof 
% of correctness of our algorithm.  This proof is based on induction and proof
% by contradiction.
\end{document}